\documentclass{article}
\usepackage{amsmath}
\usepackage{amsfonts}
\usepackage{amsthm}
\usepackage{tikz-cd}
\usepackage{tikz-feynman}

\newtheorem{theorem}{Theorem}[section]
\newtheorem{definition}[theorem]{Definition}
\newtheorem{proposition}[theorem]{Proposition}
\newtheorem{example}[theorem]{Example}
\newtheorem{lemma}[theorem]{Lemma}
\newtheorem{remark}[theorem]{Remark}

\title{An Algebraic Proof of Gauge-Fixing Independence}
\author{Jiuhe Liu}
\date{\today}
\begin{document}
\maketitle

\begin{abstract}
    This paper provides an algebraic proof of the gauge-fixing independence of
    perturbative asymptotic expansions in a finite-dimensional model of
    non-Abelian gauge theory.
    We develop an algebraic framework using the BRST formalism and perturbation theory.
    Our main theorem shows that under appropriate transversality conditions,
    normalized expectation values of observables remain unchanged under smooth
    deformations of the gauge-fixing conditions that preserve these assumptions.
\end{abstract}

\tableofcontents

\section{Introduction}\label{Introduction}

Perturbation theory and the semiclassical expansion are among the principal
tools used to interpret path integrals in physics. Feynman used perturbation
theory to develop the diagrammatic expansion for
QED~\cite{PhysRev.76.769}.
When this method is applied to non-Abelian gauge
theory, the nonlinear gauge orbit becomes a problem.
Faddeev and Popov~\cite{FADDEEV196729}
extended this method to non-Abelian gauge theories.
They defined Feynman diagrams for Yang--Mills theory
by introducing ghost Fermionic fields.
See also~\cite{PhysRev.162.1195,PhysRev.175.1580,PhysRevD.2.2841}.

This method uses a choice of a local section transverse to the gauge-group
orbits, known as a gauge condition.
Usually, the gauge condition is chosen as in 
QED by setting the divergence of the gauge field to be zero.
(This is the Lorenz gauge condition.)
Gauge-fixing independence is standard in the BRST formalism
\cite{HenneauxTeitelboim1992}, but it is not manifest term by term in the
perturbative expansion. Related finite-dimensional formulations of
gauge-fixing independence also occur in the Batalin--Vilkovisky formalism
\cite{Schwarz1993}. The purpose of this paper is to give a self-contained
algebraic argument in a finite-dimensional Faddeev--Popov--BRST model.

We prove gauge-fixing independence of the Faddeev--Popov method in a
finite-dimensional model of gauge theory.

Because the Feynman-diagram expansion is algebraic in nature, we use the BRST
formulation introduced by Becchi, Rouet, and Stora~\cite{BECCHI1976287} and by
Tyutin~\cite{tyutin2008gaugeinvariancefieldtheory}.

Thus, this paper gives a purely algebraic proof of gauge-fixing independence
for Faddeev--Popov--BRST perturbation theory in a finite-dimensional model.
Our notation follows~\cite{Reshetikhin_2010}.
The main tool of this proof is the Wick theorem.
The same identities motivate the usual formal perturbative argument in field
theory. We do not claim here to construct or justify an infinite-dimensional
path-integral measure.
Moreover, the local-section hypothesis used below deliberately excludes global
gauge-fixing obstructions such as Gribov copies~\cite{Gribov1978}.

The paper is organized as follows.
Section~\ref{Ideas of Asymptotic Series of the Integral}
contains a short overview of the FP method applied in a finite dimensional setting.
In Section~\ref{Basic Definitions and Notations}, we define the algebraic framework.
In Section~\ref{Symmetries and Algebraic Tools}, we
derive some useful identities. 
Section~\ref{Main Theorem} contains the main theorem and its proof.
Section~\ref{Example and Feynman Diagram Interpretation} contains an example
of the application of the main theorem.
Section~\ref{Appendix : Independence of the Choice of Coordinates} contains the
proof of the independence of the choice of coordinates.
Section~\ref{Appendix: Integral Interpretation of the Lemmas and Proofs} contains the
integration interpretation of the lemmas and proofs.

As an acknowledgement, I would like to thank my advisor, Professor Reshetikhin,
for his guidance and encouragement. He also gave me much helpful advice that made
this paper readable.

\section{Ideas of Asymptotic Series of the Integral}\label{Ideas of Asymptotic Series of the Integral}

The following integral is a finite dimensional analog 
of a path integral in a gauge invariant quantum field theory.

\begin{gather}\label{Action integral}
    I_h = \int_M \exp\!\left(\frac{i}{h}f(x)\right)g(x)\,dx
\end{gather}

Here $M$ is a finite dimensional manifold with an action 
of a finite dimensional compact Lie group $G$,
$f(x)$ is a $G$-invariant function on $M$,
and $g(x)\,dx$ is a $G$-invariant density. We assume that the integral is
well-defined (for example, by a compact-support hypothesis), that the action is
locally free near the critical orbit under consideration, and that $f$ has a
single nondegenerate critical orbit in that neighborhood.

We choose a basis $\{e_a\}$ of the Lie algebra $\mathfrak{g}$ of $G$
and local coordinates $x^i$ on $M$.
The Lie algebra action produces a family of vector fields $D_a = D_a^i\frac{\partial }{\partial x^i}$.
The local Faddeev--Popov formula~\cite{FADDEEV196729} (see
also~\cite{Reshetikhin_2010}) represents $I_h$, up to convention-dependent
normalizing constants in the Fourier and Berezin measures, by

\begin{gather}\label{FP integral}
    |G|\int \exp\!\left(\frac{i}{h}f(x) + c^a L_a^b(x)\bar{c}_b
    + \frac{i}{h}\lambda_a\phi^a(x)\right)g(x)\,dx\,dc\,d\bar{c}\,d\lambda .
\end{gather}

Here $|G|$ is the volume of $G$ with respect to the Haar measure,
$\phi^a(x)$ are functions on $M$ whose common zero set defines a local section
of $M\to M/G$ near the chosen critical orbit, and $L_a^b(x)$ is given by

\begin{gather*}
    L_a^b(x) =D_a^k(x)\partial_k\phi^b(x)
\end{gather*}

The assumption that the subspace $\{\phi^a(x) = 0\}$
is a local section of $M\to M/G$ implies that the matrix $(L_a^b)$
is nondegenerate.

To calculate the asymptotic expansion of~\eqref{FP integral} as
$h \to 0$, one can first find the critical point of $f(x) + \lambda_a\phi^a(x)$
(regarded as a function of $x$ and $\lambda$), assume that its Hessian is
nondegenerate, and then apply the stationary-phase method at each critical
point. For simplicity, we will only consider the case of one critical point in
this paper.

This algorithm is standard and is governed by Wick's theorem~\cite{Wick1950};
details can be found, for example, in~\cite{Reshetikhin_2010}. The
mathematical foundation of the algorithm can be found in
\cite{Grigis_Sjöstrand_1994}.

In gauge theory, the functions $\phi^a(x)$ are called gauge-fixing functions.
Their derivatives appear explicitly in the perturbative expansion.
Nevertheless, the original integral~\eqref{Action integral} does not depend on
$\phi^a(x)$. We prove the corresponding statement as an algebraic identity for
the normalized formal expansion.

\section{Basic Definitions and Notations}\label{Basic Definitions and Notations}
For a rigorous algebraic proof, we first need a precise definition of the Wick
expansion. The algorithm evaluates each monomial (equivalently, each Feynman
diagram) and then sums the results. A priori, infinitely many terms could
contribute to the coefficient of a fixed power $h^k$. We therefore do not
define the algebraic expectation on the unrestricted space generated by
$h,h^{-1},X^k,c^a,\bar c_a$. Instead, we use a completion $R$ with a finiteness
condition ensuring that every coefficient obtained after Wick contraction is
a finite sum. This completion is an inverse limit of the modules $G^k$ defined
below. It is a ring, and the exponential is defined on an appropriate filtered
subspace.

Gauge-fixing independence will be proved infinitesimally. We consider a smooth
family $\phi_t$ and require every Taylor coefficient of $\phi_t$ at $x=0$ to
depend smoothly on $t$. Accordingly, the domain of the expectation is a module
over a coefficient ring $A$, rather than merely a complex vector space. The
main application is $A=C^\infty(I,\mathbb C)$, where $I$ is an interval with
coordinate $t$. The derivation $\partial/\partial t$ then extends
coefficientwise to the modules used below.

For the constructions below, however, we initially assume only that $A$ is a commutative ring
over $\mathbb{C}$.

\begin{definition}
    For $k \in \mathbb{Z}, n,m \in \mathbb{N}_+$, we consider an $A$-module
    \begin{gather*}
        h^k A[h^{-1},X^1,\dots,X^{n+m}]\otimes_A \wedge^* A\{c^1,\dots,c^m,\bar{c}_1,\dots,\bar{c}_m\}
    \end{gather*}
    We assign degree $2$ to $h$, degree $1$ to each of
    $X^1,\dots,X^{n+m}$, and degree $0$ to the ghost variables.

    We define $G^{n+m,m,k}$ to be the submodule of $h^k
        A[h^{-1},X^1,\dots,X^{n+m}]\otimes_A \wedge^*
        A\{c^1,\dots,c^m,\bar{c}_1,\dots,\bar{c}_m\}$ generated by the monomials of
    degree not greater than $k$.
\end{definition}

For simplicity, if $n,m$ are given in the context, we use $G^k$ rather than
$G^{n+m,m,k}$.

\begin{definition}
    Let $G^{k+1}\to G^k$ be the projection that sends every monomial of degree
    $k+1$ to zero. Define
    \begin{gather*}
        R := \varprojlim_{k}G^k
    \end{gather*}

    We define the submodule $F_l \subseteq R$:
    \begin{gather*}
        F_l := \ker(R \to G^{l-1}).
    \end{gather*}

\end{definition}

Intuitively, $F_l$ consists of series whose terms have degree at least $l$.

It is easy to check the following proposition on $R$:

\begin{proposition}
    $R$ has a supercommutative ring structure.
\end{proposition}

\begin{proposition}
    For $\alpha \in F_1$, $\exp(\alpha)$ is a well-defined element in $R$.
\end{proposition}

\begin{proposition}
    $\partial_k = \frac{\partial }{\partial X^k}$ is well defined on $R$.
\end{proposition}

\begin{proposition}
    For any $\alpha ,\beta \in R , \theta \in F_1$,
    \begin{gather*}
        \partial_i(\alpha \beta) = \partial_i \alpha \beta + \alpha \partial_i \beta\\
        \partial_i\exp(\theta) = \exp(\theta)\partial_i \theta
    \end{gather*}
\end{proposition}

Motivated by a finite-dimensional Gaussian integral, we define the free
expectation by Wick's rule. Formally, it corresponds to the following
``identity'':
\begin{gather*}
    \langle g \rangle_{B,P} = \frac{\int e^{\frac{i}{h}\frac{1}{2}B_{ij} X^i X^j + P_a^b c^a\bar{c}_b} g dX^1...dX^{n+m} dc^1...dc^m d\bar{c}_1...d\bar{c}_m}{\int e^{\frac{i}{h}\frac{1}{2}B_{ij} X^i X^j + P_a^b c^a\bar{c}_b}  dX^1...dX^{n+m} dc^1...dc^m d\bar{c}_1...d\bar{c}_m}
\end{gather*}

To avoid the problem of convergence, we will not use the integral directly, but
the explicit formula of the result.
\begin{definition} (Wick theorem)

    Suppose $(B_{ij})$ is a symmetric matrix in $GL_{n+m}(A)$ and $(P_a^b)$ is a
    matrix in $GL_m(A)$. We define an $A$-module homomorphism (not a ring
    homomorphism):
    \begin{gather*}
        \langle \cdot \rangle_{k,B,P} : G^k \to h^{[k/2]}A[h^{-1}]
    \end{gather*}

    For any monomial in $G^k$, we write it in the following form:
    \begin{gather*}
        \theta = h^{\alpha}(X^1)^{\alpha_1}\dots(X^{n+m})^{\alpha_{n+m}} (c^1)^{\epsilon_1}\dots(c^m)^{\epsilon_m}(\bar{c}_1)^{\eta_1}\dots( \bar{c}_m)^{\eta_m}
    \end{gather*}

    Here $\epsilon_i,\eta_i \in \{0,1\}$. If $\sum \alpha_i$ is odd, or $\sum
        \epsilon_i \ne \sum \eta_i$, then we define $\langle \theta \rangle_{k,B,P} =
        0$.

    Otherwise we write $\theta$ in the following form:

    \begin{gather*}
        \theta = h^{\alpha}X^{s_1}X^{s_2}\dots X^{s_N}
        c^{t_1}c^{t_2}\dots c^{t_M}\bar{c}_{u_1}\bar{c}_{u_2}\dots\bar{c}_{u_M}
    \end{gather*}

    Here we allow $s_i = s_j$ for $i \ne j$, and restrict $t_1 < t_2 <\dots<t_M,
        u_1 < u_2 < \dots < u_M $. We define

    \begin{gather*}
        \langle \theta \rangle_{k,B,P} = h^\alpha (ih)^{N/2}
        \left(
        \sum_{\pi}
        (B^{-1})^{s_{\pi(1)}s_{\pi(2)}}
        \cdots
        (B^{-1})^{s_{\pi(N-1)}s_{\pi(N)}}
        \right)
        \\
        \qquad \times (-1)^{M(M-1)/2}
        \left(
        \sum_{\sigma}
        \text{sign}(\sigma)
        (P^{-1})^{t_1}_{u_{\sigma(1)}}
        (P^{-1})^{t_2}_{u_{\sigma(2)}}
        \cdots
        (P^{-1})^{t_M}_{u_{\sigma(M)}}
        \right)
    \end{gather*}

    Here $\sigma$ runs over all permutations of $(1,2,\dots,M)$, and
    $\operatorname{sign}(\sigma)=(-1)^{\operatorname{inv}(\sigma)}$. The symbol
    $\pi$ runs over all pairings of $(1,2,\dots,N)$; a pairing is a partition of
    $\{1,2,\dots,N\}$ into two-element subsets.

    Since $\theta\in G^k$, we have $2\alpha+N\le k$, and hence
    $\alpha+N/2\le\lfloor k/2\rfloor$. Thus the right-hand side belongs to
    $h^{\lfloor k/2\rfloor}A[h^{-1}]$.

    By definition, there is a commutative diagram:

    \[\begin{tikzcd}
            {G^{k+1}} & {h^{[(k+1)/2]}A[h^{-1}]} \\
            {G^{k}} & {h^{[k/2]}A[h^{-1}]}
            \arrow[from=1-1, to=1-2]
            \arrow[from=1-1, to=2-1]
            \arrow[from=1-2, to=2-2]
            \arrow[from=2-1, to=2-2]
        \end{tikzcd}\]

    Therefore it induces an $A$-module homomorphism $\langle \cdot \rangle_{B,P} : R \to
        A[[h]][h^{-1}]$.

\end{definition}

\begin{remark}
    Although $\langle\cdot\rangle_{B,P}$ is motivated by integration,
    coordinate independence is not immediate from its algebraic definition.
    The proof of this property will be delayed to section \ref{Appendix : Independence of the Choice of Coordinates}.
\end{remark}

The following example translates the stationary-phase algorithm of
Section~\ref{Ideas of Asymptotic Series of the Integral} into this algebraic
setting.

\begin{example}[]\label{exampleoffp1}
    We make the following assumptions and recover the formal expansion of~\eqref{FP integral}:

    \begin{itemize}
        \item $A = \mathbb{C}$.
        \item The local coordinates on $M$ are $x^1,\dots,x^n$. Since $\dim \mathfrak{g}=m$,
              the indices $a,b$ range over $\{1,\dots,m\}$. We also have the variables
              $\lambda_1,...,\lambda_m,c^1,...,c^m,\bar{c}_1,...,\bar{c}_m$.
        \item We work on $R^{n+m,m}$ and identify $X^{i}= x^i$ for $i = 1,\dots,n$, $X^{a+n}
                  = \lambda_a$ for $a = 1,\dots,m$. $c,\bar{c}$ are naturally identified.
        \item The unique critical point of $f(x) + \lambda_a\phi^a(x)$ is at $X = 0$.
        \item Smooth functions on $M$ are replaced by their Taylor series at $X=0$ and
              are therefore regarded as elements of $R$.
        \item The ``action'' $S$ is defined by
              \begin{gather*}
                  \frac{i}{h}S := \frac{i}{h}(f + \lambda_a \phi^a ) + c^a \bar{c}_b L_a^b.
              \end{gather*}
        \item  We consider the first few order expansion of $S$: denote
              \begin{gather*}
                  f + \lambda_a \phi^a =A_0 +  \frac{1}{2}B_{ij}X^i X^j + O(X^3)\\
                  L_a^b := D_a^k\partial_k \phi^b = P_a^b + O(X)
              \end{gather*}
              Here $D_a^k\in A[[x^1,\dots,x^n]]$, $A_0\in A$, $(P_a^b)$ is
              invertible over $A$, and $(B_{ij})$ is symmetric and invertible
              over $A$. We set $D_a^k=0$ for $k=n+1,\dots,n+m$; thus $D_a$
              acts only on the $x$ variables, not on the $\lambda$ variables.
        \item The ``interaction part of the action'' $S_{int}$ is the projection of $S$ to
              $F_3$. In other words, we have:
              \begin{gather*}
                  \frac{i}{h}S_{int} := \frac{i}{h}(f + \lambda_a \phi^a - A_0 - \frac{1}{2}B_{ij}X^i X^j )
                  +c^a\bar{c}_b(L_a^b - P_a^b).
              \end{gather*}
    \end{itemize}
    By counting the degree,
    $\frac{i}{h}S_{int}\in F_1 , g \in F_0$, so
    \begin{gather*}
        \exp(\frac{i}{h}S_{int})
    \end{gather*}
    is well defined in $F_0$, and by definition we can check that
    \begin{gather*}
        \langle \exp(\frac{i}{h}S_{int}) g\rangle_{B,P}
    \end{gather*}
    is just the asymptotic series mentioned in section \ref{Ideas of Asymptotic Series of the Integral}.
    (up to multiplying by a constant).
\end{example}

\section{Symmetries and Algebraic Tools}\label{Symmetries and Algebraic Tools}

From now on, let $\mathfrak g$ be an $m$-dimensional unimodular Lie algebra.
Choose a basis $\{T_1,\dots,T_m\}$, and write the Lie
bracket is given by the structure constants
\begin{gather*}
    [T_a,T_b] = f_{ab}^c T_c
\end{gather*}

Assume that $\mathfrak g$ acts on $A[[X^1,\dots,X^{n+m}]]$ by formal vector fields
$D_a = D_a^k\partial_k$.

\begin{definition}
    A function $f\in A[[X^1,\dots,X^{n+m}]]$ is invariant if $D_af=0$ for every $a$.

    A density $\rho\,dX^1\cdots dX^{n+m}$ is invariant if
    $\rho\in A[[X^1,\dots,X^{n+m}]]$ and, for every $a$,
    \begin{gather*}
        D_a^k\partial_k \rho + \partial_k D_a^k \rho = 0
    \end{gather*}
\end{definition}

\begin{proposition}[]\label{propofexpectation1}
    If $N$ is even, then
    \begin{gather*}
        \langle X^{s_1}\dots X^{s_N}\rangle_{B,P}
        =ih \sum_{l = 2}^N(B^{-1})^{s_1,s_l}\langle X^{s_2}\dots\widehat{X^{s_l}}\dots X^{s_N}\rangle_{B,P}
    \end{gather*}
\end{proposition}

\begin{proof}
    Group the pairings according to the index $l$ paired with $1$. The remaining
    pairs form a pairing of $\{2,\dots,\widehat l,\dots,N\}$.
\end{proof}

\begin{proposition}[]\label{propofexpectation2}
    \begin{gather*}
        \langle c^{s_1}c^{s_2}\dots c^{s_M}\bar{c}_{t_1}\bar{c}_{t_2}\dots\bar{c}_{t_M} \rangle_{B,P}
        =\sum_{l = 1}^{M}(-1)^{l+M}
        (P^{-1})^{s_l}_{t_1}
        \langle c^{s_1}\dots\widehat{c^{s_l}}\dots c^{s_M}\bar{c}_{t_2}\dots\bar{c}_{t_M} \rangle_{B,P}\\
        =\sum_{l = 1}^{M}(-1)^{l+M}
        (P^{-1})^{s_1}_{t_l}
        \langle c^{s_2}\dots c^{s_M}\bar{c}_{t_1}\dots\widehat{\bar{c}}_{t_l}\dots\bar{c}_{t_M} \rangle_{B,P}
    \end{gather*}
\end{proposition}
\begin{proof}
    This follows by expanding the determinant along its first row (respectively, column).
\end{proof}

\begin{proposition}
    Suppose $g\in h^k A[h^{-1},X^1,\dots,X^{n+m}]$,
    $u \in \wedge^* A\{c^1,\dots,c^m,\bar{c}_1,\dots,\bar{c}_m\}$,
    Then
    \begin{gather*}
        \langle g u \rangle_{B,P} = \langle g \rangle_{B,P} \langle u \rangle_{B,P}
    \end{gather*}

\end{proposition}
\begin{proof}
    It suffices to check the identity on monomials, where it follows directly from the definition.
\end{proof}

\begin{remark}
    The following proposition is the formal analog of
    \begin{gather*}
        \int \partial_j( e^{\frac{i}{h}\frac{1}{2}B_{ij} X^i X^j + P_a^b c^a\bar{c}_b} \theta
        )dX^1...dX^{n+m} dc^1...dc^m d\bar{c}_1...d\bar{c}_m = 0
    \end{gather*}
\end{remark}

\begin{proposition}[]\label{identityofpartialk}
    For any $\theta \in R$,
    \begin{gather*}
        \langle-\frac{1}{ih}B_{ij}X^i\theta + \partial_j\theta\rangle_{B,P} = 0
    \end{gather*}
\end{proposition}

\begin{proof}
    We can reduce to the case that $\theta$ is monomial,
    then we can write
    \begin{gather*}
        \theta = h^\alpha X^{s_1}\dots X^{s_N}u
    \end{gather*}
    Here $u \in \wedge^* A\{c^1,\dots,c^m,\bar{c}_1,\dots,\bar{c}_m\}$,

    Then

    \begin{align*}
         & \left\langle -\frac{1}{ih} B_{ij} X^j \theta + \partial_i^{(k)} \theta \right\rangle_{k-1, B,P} \\  & = \left\langle -\frac{1}{ih}
           B_{ij} X^j X^{s_1} \cdots X^{s_N} + \partial_i \left( X^{s_1} \cdots X^{s_N} \right) \right\rangle_{B,P} \, h^\alpha
        \left\langle u \right\rangle_{B,P}                                                                 \\  & = \left( -\frac{1}{ih} B_{ij} \cdot ih
           \sum_{l=1}^N (B^{-1})^{j s_l} \left\langle X^{s_1} \cdots \widehat{X^{s_l}}
        \cdots X^{s_N} \right\rangle_{B,P} \right.                                                         \\  & \qquad \left. + \sum_{m=1}^N
           \delta_i^{s_m} \left\langle X^{s_1} \cdots \widehat{X^{s_m}} \cdots X^{s_N}
        \right\rangle_{B,P} \right) h^\alpha \left\langle u \right\rangle_{B,P}                            \\  & =
           0
    \end{align*}

    The second equal sign uses \ref{propofexpectation1}.

\end{proof}

\begin{remark}
    The following proposition is the formal analog of
    \begin{gather*}
        \int D_a^j \partial_j( e^{\frac{i}{h}\frac{1}{2}B_{ij} X^i X^j + P_a^b c^a\bar{c}_b} \theta
        )\rho dX^1...dX^{n+m} dc^1...dc^m d\bar{c}_1...d\bar{c}_m = 0
    \end{gather*}
\end{remark}

\begin{proposition}
    If $\theta\in R$ and $\rho$ is an invariant density, then
    \begin{gather*}
        \langle  D_a^j\partial_j\theta \rho\rangle_{B,P} = \langle\frac{1}{ih}B_{ij}X^iD_a^j\theta
        \rho\rangle_{B,P}
    \end{gather*}
\end{proposition}

\begin{proof}
    \begin{gather*}
        \langle\partial_j\theta D_a^j \rho \rangle_{B,P} = \langle \partial_j(\theta D_a^j \rho) \rangle_{B,P} =\langle \frac{1}{ih }B_{ij}X^i\theta
        D_a^j\rho\rangle_{B,P}
    \end{gather*}

    The first equal sign is by definition of invariant density, the second equal
    sign is by the last proposition.

\end{proof}

Recall that the BRST differential~\cite{BECCHI1976287,tyutin2008gaugeinvariancefieldtheory}
is given by the following expression (we use the convention of~\cite{Reshetikhin_2010}):
\begin{gather*}
    c^a D_a^k \frac{\partial}{\partial X^k} -\frac{1}{2}f^a_{bc}c^b c^c \frac{\overleftarrow{\partial} }{\partial c^a} + \lambda_a \frac{\overleftarrow{\partial}}{\partial \bar{c}_a}
\end{gather*}

Motivated by the Chevalley--Eilenberg differential
\begin{gather*}
    c^a D_a^k \frac{\partial}{\partial X^k} -\frac{1}{2}f^a_{bc}c^b c^c \frac{\overleftarrow{\partial} }{\partial c^a}
\end{gather*}
we define the following operator $Q_C$:
\begin{definition}
    The odd derivation $Q_C$ on $R$ is defined on generators by

    \begin{gather*}
        Q_C(X^k) = c^a D_a^k \quad Q_C(c^a) = -\frac{1}{2}f^a_{bc}c^b c^c \quad Q_C(\bar{c}_a) = 0
    \end{gather*}
    and, for homogeneous monomials $\alpha,\beta$, by the graded Leibniz rule
    \begin{gather*}
        Q_C(\alpha \beta) = Q_C(\alpha) \beta +(-1)^{|\alpha|} \alpha Q_C(\beta)
    \end{gather*}
    Here the $\mathbb{Z}_2$ grading $|\cdot|$ is induced from the standard one on
    $\wedge^* A\{c^1,\dots,c^m,\bar{c}_1,\dots,\bar{c}_m\}$.
\end{definition}

\begin{remark}
    One can check that $Q_C^2 = 0$, but we will not use this fact.
\end{remark}

\begin{remark}
    The following proposition is the formal analog of
    \begin{gather*}
        \int e^{\frac{i}{h}\frac{1}{2}B_{ij} X^i X^j }Q_C( e^{ P_a^b c^a\bar{c}_b} u )dX^1...dX^{n+m} dc^1...dc^m d\bar{c}_1...d\bar{c}_m = 0
    \end{gather*}
\end{remark}

\begin{proposition}
    If $u\in \wedge^* A\{c^1,\dots,c^m,\bar{c}_1,\dots,\bar{c}_m\}$, then
    \begin{gather*}
        \langle Q_C(u)\rangle_{B,P} =\langle \frac{1}{2}f^a_{cd}c^c c^dP_a^b \bar{c}_b u \rangle_{B,P}
    \end{gather*}
\end{proposition}

\begin{proof}
    It suffices to take $u$ to be a monomial with one more $\bar c$ factor than
    $c$ factors; otherwise both sides vanish.

    So we write
    \begin{gather*}
        u = c^{s_1}\dots c^{s_N}\bar{c}_{t_0}\bar{c}_{t_1}\dots\bar{c}_{t_N}
    \end{gather*}
    Then
    \begin{gather*}
        \langle f^a_{cd}c^c c^dP_a^b \bar{c}_b \cdot u \rangle_{B,P}\\
        =(-1)^Nf^a_{cd}P_a^b\langle c^c c^d c^{s_1}\dots c^{s_N}\bar{c}_b\bar{c}_{t_0}\bar{c}_{t_1}\dots\bar{c}_{t_N} \rangle_{B,P}\\
        =-f^a_{cd}P_a^b(P^{-1})_b^c\langle c^dc^{s_1}\dots c^{s_N}\bar{c}_{t_0}\bar{c}_{t_1}\dots\bar{c}_{t_N}\rangle_{B,P}\\
        +f^a_{cd}P_a^b(P^{-1})_b^d\langle c^cc^{s_1}\dots c^{s_N}\bar{c}_{t_0}\bar{c}_{t_1}\dots\bar{c}_{t_N}\rangle_{B,P}\\
        +f^a_{cd}P_a^b\sum_{k=1}^{N}(-1)^k(P^{-1})_b^{s_k}
        \langle c^cc^dc^{s_1}\dots\widehat{c^{s_k}}\dots c^{s_N}\bar{c}_{t_0}\bar{c}_{t_1}\dots\bar{c}_{t_N}\rangle_{B,P}\\
        =-f_{ad}^a\langle c^d u \rangle_{B,P} + f_{ca}^a\langle c^c u \rangle_{B,P}
        + \sum_{k=1}^{N}(-1)^k f_{cd}^{s_k}
        \langle c^cc^dc^{s_1}\dots\widehat{c^{s_k}}\dots c^{s_N}\bar{c}_{t_0}\bar{c}_{t_1}\dots\bar{c}_{t_N} \rangle_{B,P}\\
    \end{gather*}

    Since $\mathfrak{g}$ is unimodular, we have
    \begin{gather*}
        f_{ad}^a = 0, f_{ca}^a = 0
    \end{gather*}

    By definition,
    \begin{gather*}
        Q_C(u) = \sum_{k=1}^{N}(-1)^k \frac{1}{2}f_{cd}^{s_k}
        c^cc^dc^{s_1}\dots\widehat{c^{s_k}}\dots c^{s_N}\bar{c}_{t_0}\bar{c}_{t_1}\dots\bar{c}_{t_N}
    \end{gather*}

    Combining these identities proves the proposition.
\end{proof}

\begin{remark}
    The following proposition is the formal analog of
    \begin{gather*}
        \int Q_C( e^{\frac{i}{h}\frac{1}{2}B_{ij} X^i X^j + P_a^b c^a\bar{c}_b} \theta )\rho dX^1...dX^{n+m} dc^1...dc^m d\bar{c}_1...d\bar{c}_m = 0
    \end{gather*}
\end{remark}

\begin{proposition}[]\label{identityofQ_C}
    For any $\theta \in R$, and any invariant density $\rho$, we have
    \begin{gather*}
        \left\langle
        \left(
        (
        -\frac{1}{ih} B_{ij} X^i c^a D_a^j
        - \frac{1}{2} f_{cd}^a c^c c^d P_a^b \bar{c}_b
        ) \theta
        + Q_C(\theta)
        \right) \rho \right\rangle_{B,P} = 0
    \end{gather*}
\end{proposition}

\begin{proof}
    By linearity it suffices to take $\theta$ to be a monomial. Write
    $\theta=gu$, where
    $g\in h^k A[h^{-1},X^1,\dots,X^{n+m}]$,
    $u \in \wedge^* A\{c^1,\dots,c^m,\bar{c}_1,\dots,\bar{c}_m\}$,
    Then
    \begin{gather*}
        \langle Q_C(\theta) \rho \rangle_{B,P}
        =\langle (D_a^j\partial_j g c^a u + g \cdot Q_C(u))\rho\rangle_{B,P}\\ =\langle D_a^j\partial_j(c^a g u) \rho\rangle_{B,P} +\langle g \rho \rangle_{B,P}\langle
        Q_C(u)\rangle_{B,P}\\ =\langle\frac{1}{ih}B_{ij}X^iD_a^j c^a g u
        \rho\rangle_{B,P} + \langle g \rho \rangle_{B,P}\langle \frac{1}{2}f^a_{cd}c^c
        c^dP_a^b \bar{c}_b \cdot u \rangle_{B,P}\\ =\langle (\frac{1}{ih} B_{ij} X^i
        c^a D_a^j + \frac{1}{2} f_{cd}^a c^c c^d P_a^b \bar{c}_b ) \theta\rangle_{B,P}
    \end{gather*}
\end{proof}

Inspired by de Rham differential
\begin{gather*}
    \lambda_a \frac{\overleftarrow{\partial}}{\partial \bar{c}_a}
\end{gather*}
we define the following operator:

\begin{definition}
    Let $M_a$ be even, ghost-independent elements of
    $A[h^{-1}][[h,X^1,\dots,X^{n+m}]]$, regarded as a subring of $R$.
    The differential operator $Q_D$ on $R$ is given by
    \begin{gather*}
        Q_D(X^k) = 0,\quad Q_D(c^a) = 0, \quad Q_D(\bar{c}_a) = M_a
    \end{gather*}
    and for monomial $\alpha,\beta$
    \begin{gather*}
        Q_D(\alpha \beta ) = Q_D(\alpha)\beta + (-1)^{|\alpha|}\alpha Q_D(\beta)
    \end{gather*}
\end{definition}

\begin{remark}
    One can check that $Q_D^2 = 0$, but we will not use this fact.
\end{remark}

\begin{remark}
    The following proposition is the formal analog of
    \begin{gather*}
        \int Q_D( e^{\frac{i}{h}\frac{1}{2}B_{ij} X^i X^j + P_a^b c^a\bar{c}_b} \theta )dX^1...dX^{n+m} dc^1...dc^m d\bar{c}_1...d\bar{c}_m = 0
    \end{gather*}
\end{remark}

\begin{proposition}[]\label{identityofQ_D}
    For any $\theta \in R$, we have
    \begin{gather*}
        \langle - c^a M_b P_a^b \theta + Q_D(\theta) \rangle_{B,P} = 0
    \end{gather*}
\end{proposition}

\begin{proof}
    By linearity it suffices to take $\theta$ to be a monomial. Write
    $\theta=gu$, where
    $g\in h^k A[h^{-1},X^1,\dots,X^{n+m}]$,
    $u \in \wedge^* A\{c^1,\dots,c^m,\bar{c}_1,\dots,\bar{c}_m\}$,

    We can write $u$ as
    \begin{gather*}
        u = c^{s_1}\dots c^{s_N}\bar{c}_{t_0}\bar{c}_{t_1}\dots\bar{c}_{t_N}
    \end{gather*}

    It suffices to consider the case in which the number of $\bar c$ factors is
    one greater than the number of $c$ factors; otherwise both sides vanish.

    We have
    \begin{align*}
        \langle c^a M_b P_a^b \theta\rangle_{B,P}
         & = \langle g M_b P_a^b c^a  c^{s_1} \dots c^{s_N} \bar{c}_{t_0} \bar{c}_{t_1} \dots \bar{c}_{t_N} \rangle_{B,P}                               \\
         & = \langle g M_b\rangle_{B,P} P_a^b
        \langle  c^a  c^{s_1} \dots c^{s_N} \bar{c}_{t_0} \bar{c}_{t_1} \dots \bar{c}_{t_N} \rangle_{B,P}                                               \\
         & = \langle g M_b\rangle_{B,P} P_a^b
        \sum_{k=0}^{N} (-1)^{N+k} (P^{-1})^a_{t_k}
        \langle c^{s_1} \dots c^{s_N} \bar{c}_{t_0} \dots \hat{\bar{c}}_{t_k} \dots \bar{c}_{t_N}\rangle_{B,P}                                          \\
         & = \sum_{k=0}^{N} \langle (-1)^{N+k} g M_{t_k} c^{s_1} \dots c^{s_N} \bar{c}_{t_0} \dots \hat{\bar{c}}_{t_k} \dots \bar{c}_{t_N}\rangle_{B,P} \\
         & = \langle g\,Q_D(u)\rangle_{B,P}
    \end{align*}
\end{proof}

\begin{example}[]\label{exampleoffp2}
    Now suppose we have a smooth family of $f,\phi,L,g$ parameterized by $t$.
    Algebraically this can be done by setting $A = C^\infty (I,\mathbb{C})$.
    For $\theta \in R$, one may think
    \begin{gather*}
        \frac{\partial}{\partial t} \langle\theta \rangle_{B,P} = \langle \frac{\partial}{\partial t} \theta \rangle_{B,P}
    \end{gather*}
    This identity need not hold because $B$ and $P$ have entries in $A$; their
    derivatives also contribute.
\end{example}

\begin{remark}
    If we write the free exponent as
    \begin{gather*}
        \mathcal E_{\mathrm{free}} := \frac{i}{h}\frac{1}{2}B_{ij} X^i X^j + P_a^b c^a\bar{c}_b,
    \end{gather*}
    The following proposition is the formal analog of
    \begin{gather*}
        \frac{\partial }{\partial t }(\frac{\sqrt{\det B}}{\det P}\int e^{\frac{i}{h}S_{free}} \theta
        dX^1...dX^{n+m} dc^1...dc^m d\bar{c}_1...d\bar{c}_m )\\ =\frac{\sqrt{\det
                B}}{\det P} \int e^{\frac{i}{h}S_{free}}(\frac{\partial }{\partial t}\theta + \frac{\partial}{\partial t}(\frac{i}{h}S_{free} )\theta)dX^1...dX^{n+m} dc^1...dc^m
        d\bar{c}_1...d\bar{c}_m \\ + \frac{\partial}{\partial t} \frac{\sqrt{\det B}}{\det P}\int e^{\frac{i}{h}\frac{1}{2}B_{ij} X^i
        X^j + P_a^b c^a\bar{c}_b} \theta dX^1...dX^{n+m} dc^1...dc^m
        d\bar{c}_1...d\bar{c}_m
    \end{gather*}
\end{remark}

\begin{proposition}\label{deformation of expectation}
    Suppose $A = C^\infty(I,\mathbb{C})$, $I$ is parameterized by $t$,
    and $\frac{\partial}{\partial t}$ is defined as the same in \ref{exampleoffp2}, then for any $\theta \in
        R$,
    \begin{gather*}
        \frac{\partial}{\partial t}\langle \theta \rangle_{B,P} =\langle \frac{\partial}{\partial t} \theta - \frac{1}{2ih} \frac{\partial}{\partial t} B_{ij}X^i X^j \theta +c^a\bar{c}_b \frac{\partial}{\partial t} P_a^b \theta \rangle_{B,P}\\ +(\frac{1}{2}\frac{\partial}{\partial t} B_{ij} (B^{-1})^{ij} - (P^{-1})^a_b\frac{\partial}{\partial t} P_a^b)\langle \theta \rangle_{B,P}\\
    \end{gather*}

\end{proposition}

\begin{remark}
    Equivalently, the right-hand side is the covariance formula
    \begin{gather*}
        \left\langle \frac{\partial\theta}{\partial t}
        + \frac{\partial\mathcal E_{\mathrm{free}}}{\partial t}\theta \right\rangle_{B,P}
        -\left\langle \frac{\partial\mathcal E_{\mathrm{free}}}{\partial t}\right\rangle_{B,P}
        \langle \theta \rangle_{B,P}.
    \end{gather*}
\end{remark}

\begin{proof}
    We first calculate some special but useful cases:

    First suppose that $M$ is even. Then
    \begin{gather*}
        \frac{\partial}{\partial t} \langle X^{s_1}\dots X^{s_M}\rangle_{B,P} =\frac{\partial}{\partial t} ( (ih)^{M/2}\sum_{\pi} (B^{-1})^{s_{\pi(1)}s_{\pi(2)}}\dots
        (B^{-1})^{s_{\pi(M-1)}s_{\pi(M)}} )\\ =(ih)^{M/2}\sum_{\pi} \sum_{k = 1}^{M/2}
        (B^{-1})^{s_{\pi(1)}s_{\pi(2)}}\dots \frac{\partial}{\partial t} (B^{-1})^{s_{\pi(2k-1)}s_{\pi(2k)}}\dots
        (B^{-1})^{s_{\pi(M-1)}s_{\pi(M)}}\\ =(ih)^{M/2}\sum_{\pi} \sum_{k = 1}^{M/2}
        (B^{-1})^{s_{\pi(1)}s_{\pi(2)}}\dots (-1) (B^{-1})^{s_{\pi(2k-1)} u}\frac{\partial}{\partial t} B_{uv} (B^{-1})^{v s_{\pi(2k)}}\dots
        (B^{-1})^{s_{\pi(M-1)}s_{\pi(M)}}\\ =(ih)^{M/2}(-\frac{\partial}{\partial t} B_{uv}) \sum_{\pi} \sum_{k = 1}^{M/2}
        (B^{-1})^{s_{\pi(1)}s_{\pi(2)}}\dots (B^{-1})^{s_{\pi(2k-1)} u} (B^{-1})^{v
                s_{\pi(2k)}}\dots (B^{-1})^{s_{\pi(M-1)}s_{\pi(M)}}\\ =(ih)^{M/2}(-\frac{\partial}{\partial t} B_{uv}) \frac{1}{2}\sum_{\pi} \sum_{k = 1}^{M/2} (
        (B^{-1})^{s_{\pi(1)}s_{\pi(2)}}\dots(B^{-1})^{s_{\pi(2k-1)} u} (B^{-1})^{v
                s_{\pi(2k)}}\dots (B^{-1})^{s_{\pi(M-1)}s_{\pi(M)}}\\
        +(B^{-1})^{s_{\pi(1)}s_{\pi(2)}}\dots(B^{-1})^{s_{\pi(2k-1)} v} (B^{-1})^{u
                s_{\pi(2k)}}\dots (B^{-1})^{s_{\pi(M-1)}s_{\pi(M)}} )
    \end{gather*}

    Now let's see what the summation contains. Suppose $s_{M+1} = u,s_{M+2} = v$,
    then the summation is over all pairing of $\{1,2,\dots,M,M+1,M+2\}$, except for
    the pairing in which $M+1$ and $M+2$ are paired together. Consequently,

    \begin{gather*}
        \langle X^u X^v X^{s_1}\dots X^{s_M}\rangle_{B,P} - \langle X^v X^u\rangle_{B,P} \langle X^{s_1}\dots X^{s_M}\rangle_{B,P}\\
        =(ih)^{M/2 +1} \sum_{\pi} \sum_{k = 1}^{M/2}
        ((B^{-1})^{s_{\pi(1)}s_{\pi(2)}}\dots (B^{-1})^{s_{\pi(2k-1)} u} (B^{-1})^{v s_{\pi(2k)}}\dots
        (B^{-1})^{s_{\pi(M-1)}s_{\pi(M)}}\\
        +(B^{-1})^{s_{\pi(1)}s_{\pi(2)}}\dots (B^{-1})^{s_{\pi(2k-1)} v} (B^{-1})^{u s_{\pi(2k)}}\dots
        (B^{-1})^{s_{\pi(M-1)}s_{\pi(M)}}
        )
    \end{gather*}

    Hence, if $g=X^{s_1}\dots X^{s_M}$ and $M$ is even, then
    \begin{gather*}
        \frac{\partial}{\partial t} \langle g\rangle_{B,P} =\langle - \frac{1}{ih} \frac{1}{2}\frac{\partial}{\partial t} B_{uv}X^u X^v g\rangle_{B,P} +\langle \frac{1}{ih} \frac{1}{2} \frac{\partial}{\partial t} B_{uv}X^u X^v\rangle_{B,P} \langle g\rangle_{B,P}
    \end{gather*}

    The identity also holds for odd $M$, because both sides then vanish.

    We also consider
    \begin{gather*}
        \frac{\partial}{\partial t} \langle c^{t_1}c^{t_2}\dots
        c^{t_N}\bar{c}_{u_1}\bar{c}_{u_2}\dots\bar{c}_{u_N}\rangle_{B,P} =\frac{\partial}{\partial t}( (-1)^{N(N-1)/2} \sum_{\sigma} \text{sign}(\sigma)
        (P^{-1})^{t_1}_{u_{\sigma(1)}} (P^{-1})^{t_2}_{u_{\sigma(2)}} \dots
        (P^{-1})^{t_N}_{u_{\sigma(N)}} )\\ =(-1)^{N(N-1)/2}\sum_{\sigma}\sum_{k =
            1}^{N}\text{sign}(\sigma) (P^{-1})^{t_1}_{u_{\sigma(1)}} \dots \frac{\partial}{\partial t} (P^{-1})^{t_k}_{u_{\sigma(k)}} \dots (P^{-1})^{t_N}_{u_{\sigma(N)}}\\
        =(-1)^{N(N-1)/2}\sum_{\sigma}\sum_{k = 1}^{N}\text{sign}(\sigma)
        (P^{-1})^{t_1}_{u_{\sigma(1)}} \dots (-1)(P^{-1})^{t_k}_{a} \frac{\partial}{\partial t} P^a_b (P^{-1})^{b}_{u_{\sigma(k)}} \dots
        (P^{-1})^{t_N}_{u_{\sigma(N)}}\\ =(-1)^{N(N-1)/2}(-\frac{\partial}{\partial t} P^a_b)\sum_{\sigma}\sum_{k = 1}^{N}\text{sign}(\sigma)
        (P^{-1})^{t_1}_{u_{\sigma(1)}} \dots (P^{-1})^{t_k}_{a}
        (P^{-1})^{b}_{u_{\sigma(k)}} \dots (P^{-1})^{t_N}_{u_{\sigma(N)}}
    \end{gather*}

    Now let's see what the summation contains. Suppose $t_{N+1} = a,u_{N+2} = b$,
    then the summation is over all permutation of $\{1,2,\dots,N,N+1\}$, except for
    the permutations for which $\sigma(N+1)=N+1$. It follows that

    \begin{gather*}
        \frac{\partial}{\partial t}\langle c^{t_1}c^{t_2}\dots
        c^{t_N}\bar{c}_{u_1}\bar{c}_{u_2}\dots\bar{c}_{u_N}\rangle_{B,P}\\ =\langle
        \frac{\partial}{\partial t} P_a^b c^a \bar{c}_b c^{t_1}c^{t_2}\dots
        c^{t_N}\bar{c}_{u_1}\bar{c}_{u_2}\dots\bar{c}_{u_N} \rangle_{B,P} - \langle
        \frac{\partial}{\partial t} P_a^b c^a \bar{c}_b \rangle_{B,P} \langle c^{t_1}c^{t_2}\dots
        c^{t_N}\bar{c}_{u_1}\bar{c}_{u_2}\dots\bar{c}_{u_N} \rangle_{B,P}
    \end{gather*}

    Thus, for $u = c^{t_1}c^{t_2}\dots
        c^{t_N}\bar{c}_{u_1}\bar{c}_{u_2}\dots\bar{c}_{u_{N'}}$,\ if $N = N'$, we have
    \begin{gather*}
        \frac{\partial}{\partial t}\langle u\rangle_{B,P} =\langle \frac{\partial}{\partial t} P_a^b c^a \bar{c}_b u \rangle_{B,P} - \langle \frac{\partial}{\partial t} P_a^b c^a \bar{c}_b \rangle_{B,P} \langle u \rangle_{B,P}
    \end{gather*}

    If $N \ne N'$, the equation is also true, since both sides are $0$.

    Now suppose $\theta \in R$ is monomial, then we write
    \begin{gather*}
        \theta = h^\alpha \lambda g u
    \end{gather*}
    Here $\lambda \in A, u = c^{t_1}c^{t_2}\dots c^{t_N}\bar{c}_{u_1}\bar{c}_{u_2}\dots\bar{c}_{u_{N'}}$,
    and $g = X^{s_1}\dots X^{s_M}$.

    So
    \begin{gather*}
        \frac{\partial}{\partial t} \langle \theta \rangle_{B,P} = \frac{\partial}{\partial t} (h^\alpha \lambda \langle g \rangle_{B,P} \langle u \rangle_{B,P} ) \\
        = h^\alpha \frac{\partial}{\partial t} \lambda \langle g \rangle_{B,P} \langle u \rangle_{B,P} + h^\alpha
        \lambda \frac{\partial}{\partial t} \langle g \rangle_{B,P} \langle u \rangle_{B,P} + h^\alpha \lambda
        \langle g \rangle_{B,P} \frac{\partial}{\partial t} \langle u \rangle_{B,P} \\ = h^\alpha \frac{\partial}{\partial t} \lambda \langle g \rangle_{B,P} \langle u \rangle_{B,P} + h^\alpha
        \lambda (\langle - \frac{1}{ih} \frac{1}{2}\frac{\partial}{\partial t} B_{ij}X^i X^j g\rangle_{B,P} +\langle \frac{1}{ih} \frac{1}{2} \frac{\partial}{\partial t} B_{ij}X^i X^j\rangle_{B,P} \langle g\rangle_{B,P} ) \langle u
        \rangle_{B,P}\\ + h^\alpha \lambda \langle g \rangle_{B,P} (\langle \frac{\partial}{\partial t} P_a^b c^a \bar{c}_b u \rangle_{B,P} - \langle \frac{\partial}{\partial t} P_a^b c^a \bar{c}_b \rangle_{B,P} \langle u \rangle_{B,P})\\ =\langle
        \frac{\partial}{\partial t} \theta - \frac{1}{ih} \frac{1}{2}\frac{\partial}{\partial t} B_{ij}X^i X^j \theta +c^a\bar{c}_b \frac{\partial}{\partial t} P_a^b \theta \rangle_{B,P}\\ +\langle \frac{1}{ih} \frac{1}{2} \frac{\partial}{\partial t} B_{ij}X^i X^j -c^a\bar{c}_b \frac{\partial}{\partial t} P_a^b \rangle_{B,P}\langle \theta \rangle_{B,P}\\ =\langle \frac{\partial}{\partial t} \theta - \frac{1}{ih} \frac{1}{2}\frac{\partial}{\partial t} B_{ij}X^i X^j \theta +c^a\bar{c}_b \frac{\partial}{\partial t} P_a^b \theta \rangle_{B,P}\\ +(\frac{1}{2}\frac{\partial}{\partial t} B_{ij} (B^{-1})^{ij} - (P^{-1})^a_b\frac{\partial}{\partial t} P_a^b)\langle \theta \rangle_{B,P}
    \end{gather*}

    By linearity and continuity in the filtration, the identity holds for every $\theta\in R$.

\end{proof}

\section{Main Theorem}\label{Main Theorem}
We now collect the hypotheses and prove the main theorem.

\begin{theorem}[]\label{algebraicindependence}
    (purely algebraic version)

    Adopt the notation of Example~\ref{exampleoffp1}, with the following additional assumptions:
    \begin{itemize}
        \item All data depend smoothly on $t$; equivalently,
              $A=C^{\infty}(I,\mathbb{C})$, where $I$ is an open interval with coordinate $t$.
        \item The formal vector fields $D_a^k \partial_k$ define an action of the unimodular $m$-dimensional Lie
              algebra $\mathfrak{g}$ with structure constants $f_{ab}^c$.
        \item $f$ is invariant and $g$ is an invariant density.
        \item There exists $W^k \in A[[x^1,...,x^n]]$ for $k = 1,\dots,n$ describing the
              the deformation induced by the chosen coordinates; namely,
              \begin{gather*}
                  \frac{\partial}{\partial t} f = W^k\partial_k f, \quad \frac{\partial}{\partial t} g = \partial_k(W^k g), \quad \frac{\partial}{\partial t} D_a^l = W^k\partial_k D_a^l - D_a^k\partial_k W^l
              \end{gather*}
              We also set $W^k = 0$ for $k = n+1,...,n+m$.
    \end{itemize}

    Then
    \begin{gather*}
        \frac{\partial}{\partial t} \langle \exp(\frac{i}{h}S_{int}) g \rangle_{B,P} = (\frac{1}{2}\frac{\partial}{\partial t} B_{ij} (B^{-1})^{ij} - (P^{-1})^a_b\frac{\partial}{\partial t} P_a^b)\langle \exp(\frac{i}{h}S_{int}) g \rangle_{B,P}
    \end{gather*}

    Equivalently, after choosing a smooth square root of $\det B$ on $I$,
    \begin{gather*}
        \frac{\partial}{\partial t}(\frac{\det P}{\sqrt{\det B}} \langle \exp(\frac{i}{h}S_{int}) g
        \rangle_{B,P}) = 0
    \end{gather*}
\end{theorem}

\begin{remark}
    Apart from preservation of the stated invertibility and smoothness
    assumptions, no restriction is imposed on the deformation of the
    gauge-fixing functions $\phi^a$; their variation cancels from the normalized expression.
\end{remark}

\begin{remark}
    The following lemma is the formal analog of
    \begin{gather*}
        \int Q_C(\exp(\frac{i}{h}S)\theta') g dX^1...dX^{n+m} dc^1...dc^m d\bar{c}_1...d\bar{c}_m  = 0
    \end{gather*}
\end{remark}

\begin{lemma}\label{Q_C lemma}
    For any $\theta' \in R$ and any invariant density $g$, we have
    \begin{gather*}
        \langle \exp(\frac{i}{h}S_{int})(Q_C(\frac{i}{h}S)\theta' + Q_C(\theta') ) g\rangle_{B,P} = 0
    \end{gather*}
\end{lemma}

\begin{proof}
    For any $\theta' \in R$, let
    \begin{gather*}
        \theta = \exp(\frac{i}{h}S_{int})\theta' \quad \rho = g
    \end{gather*}
    Then
    \begin{gather*}
        Q_C(\theta) = \exp(\frac{i}{h}S_{int}) (Q_C(\frac{i}{h}S_{int})\theta' + Q_C(\theta'))\\
        = \exp(\frac{i}{h}S_{int})
        (Q_C(\frac{i}{h}(f + \lambda_a \phi^a - A_0 - \frac{1}{2}B_{ij}X^i X^j )
        +c^a\bar{c}_b(L_a^b - P_a^b))\theta' + Q_C(\theta'))\\
        = \exp(\frac{i}{h}S_{int}) ((Q_C(\frac{i}{h}S)
        +\frac{1}{ih}B_{ij}X^i c^a D_a^j + \frac{1}{2}f^a_{cd}c^c c^d P_a^b \bar{c}_b
        )\theta' + Q_C(\theta'))
    \end{gather*}
    Proposition~\ref{identityofQ_C} therefore gives
    \begin{gather*}
        0 = \langle ((
        -\frac{1}{ih} B_{ij} X^i c^a D_a^j
        - \frac{1}{2} f_{cd}^a c^c c^d P_a^b \bar{c}_b
        ) \theta
        + Q_C(\theta)) \rho \rangle_{B,P}\\
        =\langle (\exp(\frac{i}{h}S_{int})( -\frac{1}{ih} B_{ij} X^i c^a D_a^j
        - \frac{1}{2} f_{cd}^a c^c c^d P_a^b \bar{c}_b)\theta'\\
        +\exp(\frac{i}{h}S_{int}) ((Q_C(\frac{i}{h}S)
        +\frac{1}{ih}B_{ij}X^i c^a D_a^j + \frac{1}{2}f^a_{cd}c^c c^d P_a^b \bar{c}_b
        )\theta' + Q_C(\theta')))g\rangle_{B,P}\\
        =\langle \exp(\frac{i}{h}S_{int})(Q_C(\frac{i}{h}S)\theta' + Q_C(\theta') ) g\rangle_{B,P}
    \end{gather*}
\end{proof}
\begin{remark}
    The following lemma is the formal analog of
    \begin{gather*}
        \int Q_D(\exp(\frac{i}{h}S)\theta') g dX^1...dX^{n+m} dc^1...dc^m d\bar{c}_1...d\bar{c}_m  = 0
    \end{gather*}
\end{remark}
\begin{lemma}\label{Q_D lemma}
    For any $\theta' \in R$ and any invariant density $g$, we have
    \begin{gather*}
        \langle \exp(\frac{i}{h}S_{int})(Q_D(\frac{i}{h}S)\theta' + Q_D(\theta') ) g\rangle_{B,P} = 0
    \end{gather*}
\end{lemma}

\begin{proof}
    For any $\theta' \in R$, let
    \begin{gather*}
        \theta = \exp(\frac{i}{h}S_{int})\theta' g
    \end{gather*}
    Then
    \begin{gather*}
        Q_D(\theta) = \exp(\frac{i}{h}S_{int}) (Q_D(\frac{i}{h}S_{int})\theta' + Q_D(\theta'))g\\
        = \exp(\frac{i}{h}S_{int})
        (Q_D(\frac{i}{h}(f + \lambda_a \phi^a - A_0 - \frac{1}{2}B_{ij}X^i X^j )
        +c^a\bar{c}_b(L_a^b - P_a^b))\theta' + Q_D(\theta'))g\\
        = \exp(\frac{i}{h}S_{int}) ((Q_D(\frac{i}{h}S) + c^a M_b P_a^b)\theta' + Q_D(\theta'))g
    \end{gather*}
    Proposition~\ref{identityofQ_D} therefore gives
    \begin{gather*}
        0 = \langle -c^a M_b P_a^b \theta
        + Q_D(\theta)\rangle_{B,P}\\
        =\langle  \exp(\frac{i}{h}S_{int})(-c^a M_b P_a^b)\theta' g
        + \exp(\frac{i}{h}S_{int}) ((Q_D(\frac{i}{h}S) + c^a M_b P_a^b)\theta' + Q_D(\theta'))g\rangle_{B,P}\\
        = \langle \exp(\frac{i}{h}S_{int})(Q_D(\frac{i}{h}S)\theta' + Q_D(\theta') ) g\rangle_{B,P}
    \end{gather*}
\end{proof}

\begin{remark}
    The following lemma is the formal analog of
    \begin{gather*}
        \int \partial_k(\exp(\frac{i}{h}S)\theta') dX^1...dX^{n+m} dc^1...dc^m
        d\bar{c}_1...d\bar{c}_m = 0
    \end{gather*}
\end{remark}

\begin{lemma}
    For any $\theta' \in R$, we have
    \begin{gather*}
        \langle \exp(\frac{i}{h}S_{int})(\partial_k(\frac{i}{h}S)\theta' + \partial_k(\theta') ) \rangle_{B,P} = 0
    \end{gather*}
\end{lemma}

\begin{proof}
    By \ref{identityofpartialk}, we have
    \begin{gather*}
        0 = \langle \frac{i}{h}B_{jk}X^j
        \exp(\frac{i}{h}S_{int})\theta'
        +\partial_k(\exp(\frac{i}{h}S_{int})\theta')\rangle_{B,P}\\ =\langle
        \frac{i}{h}B_{jk}X^j \exp(\frac{i}{h}S_{int})\theta'\\
        +\exp(\frac{i}{h}S_{int})(\partial_k (\frac{i}{h}S) + \partial_k ( - A_0 -\frac{1}{2}\frac{i}{h}B_{ij}X^i X^j - c^a \bar{c}_b
        P_a^b))\theta' +\exp(\frac{i}{h}S_{int})\partial_k(\theta')\rangle_{B,P}\\ =\langle \exp(\frac{i}{h}S_{int})(\partial_k(\frac{i}{h}S)\theta' + \partial_k(\theta') ) \rangle_{B,P}
    \end{gather*}
\end{proof}

From now on we set
\begin{gather*}
    M_a = \frac{i}{h}\lambda_a
\end{gather*}

We now record the key identity.
\begin{remark}
    With the BRST differential $Q=Q_C+Q_D$, the following lemma is the formal analogue of
    \begin{gather*}
        \int \exp(\frac{i}{h}S)Q(\theta') g dX^1...dX^{n+m} dc^1...dc^m d\bar{c}_1...d\bar{c}_m   = 0
    \end{gather*}
\end{remark}

\begin{lemma}
    For any $\theta' \in R$ and any invariant density $g$, we have
    \begin{gather*}
        \langle \exp(\frac{i}{h}S_{int})(Q_C(\theta') + Q_D(\theta') ) g\rangle_{B,P} = 0
    \end{gather*}
\end{lemma}

\begin{proof}
    Notice
    \begin{gather*}
        Q_C(\frac{i}{h}S)
        =\frac{i}{h} c^a D_a^k\partial_k f + \frac{i}{h} \lambda_a c^b D_b^k\partial_k \phi^a -\frac{1}{2}f^a_{cd}c^c c^d \bar{c}_b D_a^k\partial_k \phi^b +c^a\bar{c}_b c^c D_c^l\partial_l D_a^k\partial_k \phi^b\\ = \frac{i}{h} \lambda_a c^b D_b^k\partial_k \phi^a = -Q_D(\frac{i}{h}S)
    \end{gather*}

    Adding the preceding two lemmas gives
    \begin{gather*}
        0 =\langle \exp(\frac{i}{h}S_{int})(Q_C(\frac{i}{h}S)\theta' + Q_C(\theta') + Q_D(\frac{i}{h}S)\theta' + Q_D(\theta') ) g\rangle_{B,P} \\
        =\langle \exp(\frac{i}{h}S_{int})(Q_C(\theta') + Q_D(\theta') ) g\rangle_{B,P}
    \end{gather*}
\end{proof}

Now we are ready to prove the main theorem.

\begin{remark}
    The integral interpretation gives a more intuitive version of the argument;
    it is recorded in the appendix as a formal heuristic.
\end{remark}

\begin{proof}
    Let
    \begin{gather*}
        u^a = \frac{\partial}{\partial t} \phi^a - W^k\partial_k\phi^a
    \end{gather*}
    Then

    \begin{gather*}
        \frac{\partial}{\partial t} \langle \exp(\frac{i}{h}S_{int}) g \rangle_{B,P} - (\frac{1}{2}\frac{\partial}{\partial t} B_{ij} (B^{-1})^{ij} - (P^{-1})^a_b\frac{\partial}{\partial t} P_a^b)\langle \exp(\frac{i}{h}S_{int}) g \rangle_{B,P}\\ =\langle
        \frac{\partial}{\partial t} (\exp(\frac{i}{h}S_{int}) g) - \frac{1}{ih} \frac{1}{2}\frac{\partial}{\partial t} B_{ij}X^i X^j \exp(\frac{i}{h}S_{int}) g +c^a\bar{c}_b \frac{\partial}{\partial t} P_a^b \exp(\frac{i}{h}S_{int}) g \rangle_{B,P}\\ =\langle
        \exp(\frac{i}{h}S_{int}) (\frac{\partial}{\partial t} (\frac{i}{h}S) g + \frac{\partial}{\partial t} g)\rangle_{B,P}\\ =\langle \exp(\frac{i}{h}S_{int})(\frac{\partial}{\partial t}(\frac{i}{h}(f + \lambda_a\phi^a) + c^a\bar{c}_b D_a^l \partial_l \phi^b) g + \frac{\partial}{\partial t} g)\rangle_{B,P}\\ =\langle \exp(\frac{i}{h}S_{int}) ((\frac{i}{h}(W^k\partial_k f + \lambda_a W^k \partial_k \phi^a + \lambda_a u^a ) \\ + c^a \bar{c}_b ((W^k\partial_k D_a^l- D_a^k\partial_k W^l)\partial_l \phi^b +D_a^l\partial_l( W^k\partial_k \phi^b + u^b)))g + \partial_k(W^k g))\rangle_{B,P}\\ =\langle \exp(\frac{i}{h}S_{int})
        ((\frac{i}{h}(W^k\partial_k f + \lambda_a W^k \partial_k \phi^a) + c^a \bar{c}_b (W^k\partial_k D_a^l \partial_l \phi^b +D_a^l\partial_l( W^k\partial_k \phi^b )))g + \partial_k(W^k g))\rangle_{B,P}\\ +\langle \exp(\frac{i}{h}S_{int})
        (\frac{i}{h}\lambda_a u^a + c^a \bar{c}_b D_a^l\partial_l u^b)g\rangle_{B,P}\\ =\langle \exp(\frac{i}{h}S_{int})(\partial_k(\frac{i}{h}S)W^k g + \partial_k(W^k g))\rangle_{B,P} +\langle \exp(\frac{i}{h}S_{int}) (Q_D(\bar{c}_a u
        ^a) + Q_C(\bar{c}_a u ^a))g\rangle_{B,P}\\ =0
    \end{gather*}

    That is,
    \begin{gather*}
        \frac{\partial}{\partial t}(\frac{\det P}{\sqrt{\det B}} \langle \exp(\frac{i}{h}S_{int}) g
        \rangle_{B,P}) = 0
    \end{gather*}
\end{proof}

Now we translate this theorem into the language of smooth function, and write
it in a more convenient way. We give a simple case first. In this case, we
deform the gauge-fixing function $\phi$ while the saddle point remains at $0$.
All functions are assumed to be sufficiently smooth.

\begin{theorem}\label{deformation theorem 1}
    If the following conditions are satisfied:
    \begin{itemize}
        \item $U$ is an open subset of $\mathbb{R}^n$, with coordinates $x^1,\dots,x^n$, and $0\in U$.
        \item There are $m$ vector fields $D_a=D_a^k\frac{\partial}{\partial x^k}$,
              $a=1,\dots,m$, defining an action of a unimodular Lie algebra
              $\mathfrak{g}$ with structure constants $f_{ab}^c$ :
              \begin{gather*}
                  [D_a,D_b] = f_{ab}^c D_c
              \end{gather*}
        \item $f:U\to\mathbb{R}$ is invariant under $\mathfrak g$, and
              $g,g_0:U\to\mathbb{R}$ represent invariant densities; explicitly,
              i.e.
              \begin{gather*}
                  D_a^k\frac{\partial}{\partial x^k} f = D_a^k\frac{\partial}{\partial x^k} g +\frac{\partial}{\partial x^k} D_a^k g = D_a^k\frac{\partial}{\partial x^k} g_0 +\frac{\partial}{\partial x^k} D_a^k g_0 = 0
              \end{gather*}
        \item $I$ is an open interval with coordinate $t$, and
              $\phi^a:I\times U\to\mathbb{R}$, $a=1,\dots,m$, is smooth. Write
              $\phi_t^a:U\to\mathbb{R}$ for its restriction at $t$.
        \item For every $t\in I$, $\phi_t(0)=0$ and $df(0)=0$.
        \item The following (block) matrix
              \begin{gather*}
                  \begin{pmatrix}
                      \frac{\partial^2 f}{\partial x^i \partial x^j} & \frac{\partial \phi^a_t}{\partial x^i} \\ \frac{\partial \phi^b_t}{\partial x^j} & 0
                  \end{pmatrix},
                  \begin{pmatrix}
                      D_a^k\frac{\partial \phi^b_t}{ \partial x^k}
                  \end{pmatrix}
              \end{gather*}
              are invertible at $x=0$ for all $t$.
        \item $g_0(0)\ne0$.
        \item $f,g,g_0,D_a^k , \phi^a$ are identified as their Taylor series at the point $x^k = 0$.
    \end{itemize}

    \begin{gather*}
        \frac{\partial}{\partial t} (\frac{\langle \exp(\frac{i}{h}S_{int}) g \rangle_{B,P}}{\langle
            \exp(\frac{i}{h}S_{int})g_0 \rangle_{B,P}}) = 0
    \end{gather*}

\end{theorem}

\begin{proof}
    It remains to verify the hypotheses of Theorem~\ref{algebraicindependence}.
    The invertibility conditions on $B$ and $P$ follow from the two matrix hypotheses.
    Since $f,g,D_a^k$ are independent of $t$,
    we can take $W^k = 0$.
    By definition of $\exp$,
    \begin{gather*}
        \exp(\frac{i}{h}S_{int}) = 1 + \text{elements in }F_1
    \end{gather*}
    So
    \begin{gather*}
        \langle \exp(\frac{i}{h}S_{int})g_0 \rangle_{B,P} = g_0(0) + O(h)
    \end{gather*}
    is invertible in $A[[h]]$ because $g_0(0)$ is nowhere zero on the interval
    under consideration. The result now follows from Theorem~\ref{algebraicindependence}.
    \begin{gather*}
        \frac{\partial}{\partial t} (\frac{\langle \exp(\frac{i}{h}S_{int}) g \rangle_{B,P}}{\langle
            \exp(\frac{i}{h}S_{int})g_0 \rangle_{B,P}}) =0
    \end{gather*}
\end{proof}

Now we come to the general case: when we deform the gauge fixing term $\phi$,
the saddle point $\tilde{y}$ moves, so the Taylor series of $f,g,D$ also
deform.

\begin{theorem}

    If the following conditions are satisfied:
    \begin{itemize}
        \item $U$ is an open set in $\mathbb{R}^n$ (with coordinate $y^1,...,y^n$ ).
        \item There are $m$ vector fields $D_a=D_a^k\frac{\partial}{\partial y^k}$,
              $a=1,\dots,m$, defining an action of a unimodular Lie algebra
              $\mathfrak{g}$ with structure constants $f_{ab}^c$ :
              \begin{gather*}
                  [D_a,D_b] = f_{ab}^c D_c
              \end{gather*}
        \item $f: U \to \mathbb{R}$ is invariant function under $\mathfrak{g}$,
              $g,g_0: U \to \mathbb{R}$ are invariant density under $\mathfrak{g}$,
              i.e.
              \begin{gather*}
                  D_a^k\frac{\partial f}{\partial y^k} =
                  D_a^k\frac{\partial g}{\partial y^k} +\frac{\partial D_a^k}{\partial y^k}g =
                  D_a^k\frac{\partial g_0}{\partial y^k} +\frac{\partial D_a^k}{\partial y^k}g_0 = 0
              \end{gather*}
        \item $I$ is an open interval parameterized by $t$, and there exists $\phi^a
                  : I \times U \to \mathbb{R}, a =1,\dots,m$. Denote $\phi_t^a : U \to \mathbb{R}$ to be
              the restriction of $\phi^a$ to $U$ at $t$.
        \item There exists $\tilde{y} : I \to U , t\mapsto \tilde{y}_t$. Write in
              coordinates, we have $\tilde{y}^k : I \to \mathbb{R} , t\mapsto \tilde{y}^k_t$.
        \item For every $t\in I$, $\phi_t(\tilde y_t)=0$ and $df(\tilde y_t)=0$.
        \item The following (block) matrix
              \begin{gather*}
                  \begin{pmatrix}
                      \frac{\partial^2 f}{\partial y^i \partial y^j} & \frac{\partial \phi^a_t}{\partial y^i} \\ \frac{\partial \phi^b_t}{\partial y^j} & 0
                  \end{pmatrix},
                  \begin{pmatrix}
                      D_a^k\frac{\partial \phi^b_t}{ \partial y^k}
                  \end{pmatrix}
              \end{gather*}
              are invertible at $y = \tilde{y}_t$ for all $t$.
        \item $g_0(\tilde y_t)\ne0$ for every $t\in I$.
        \item $f,g,g_0,D_a^k , \phi^a$ are identified with their Taylor series at the point $y^k = \tilde{y}^k_t$,
              equivalently, at $x^k=y^k-\tilde y_t^k=0$.
    \end{itemize}

    Then
    \begin{gather*}
        \frac{\partial}{\partial t} (\frac{\langle \exp(\frac{i}{h}S_{int}) g \rangle_{B,P}}{\langle
            \exp(\frac{i}{h}S_{int}) g_0\rangle_{B,P}}) = 0
    \end{gather*}
\end{theorem}

\begin{proof}
    The hypotheses of Theorem~\ref{algebraicindependence} are satisfied. The
    proof is as above, except that $W^k$ is no longer zero.
    Since for $f$ identified as element in $A[[X^1,\dots,X^{n+m}]]$, we have
    \begin{gather*}
        \frac{\partial}{\partial t} f = \sum_{N\ge 0}\frac{1}{N!}\frac{\partial }{\partial t} (\frac{\partial^N f}{\partial y^{i_1}\dots\partial y^{i_N}}(\tilde{y}_t) ) x^{i_1}\dots x^{i_N}\\ = \sum_{N\ge
            0}\frac{1}{N!}\frac{\partial \tilde{y}^k_t}{\partial t} \frac{\partial^{N+1} f}{\partial y^k\partial y^{i_1}\dots\partial y^{i_N}}(\tilde{y}_t)x^{i_1}\dots x^{i_N}\\ =\frac{\partial \tilde{y}^k_t}{\partial t} \sum_{N\ge 0}\frac{1}{N!} \frac{\partial^{N+1} f}{\partial y^k\partial y^{i_1}\dots\partial y^{i_N}}(\tilde{y}_t)x^{i_1}\dots x^{i_N}\\ =\frac{\partial \tilde{y}^k_t}{\partial t} \frac{\partial}{\partial x^k} f
    \end{gather*}

    The corresponding pullback identities hold for the density and vector
    fields (including the divergence and commutator terms in
    Theorem~\ref{algebraicindependence}). It therefore suffices to take
    \begin{gather*}
        W^k = \frac{\partial \tilde{y}^k_t}{\partial t}
    \end{gather*}

    The preceding argument now applies.
\end{proof}

\section{Example and Feynman-Diagram Interpretation}\label{Example and Feynman Diagram Interpretation}

The calculation in the proof can be easily translated into the language of
Feynman diagrams. This translation indicates how the algebraic identities are
used in formal perturbative field theory.

To keep the calculation simple, we remain in the finite-dimensional setting
and illustrate the diagrammatic translation there.

Now suppose $M = \mathbb{R}^2 - \{0\}$, with (local) coordinate $(r,\theta)$.
The $U(1)$-action on $M$ is generated by the vector field $\partial_\theta$; thus
\begin{gather*}
    D^r = 0, D^\theta = 1
\end{gather*}
Let $f$ be an invariant function on $M$ and $g$ an invariant density. In these
local coordinates, their coefficient functions depend only on $r$.

Suppose that one critical orbit of $f$ is $r=r_0$ and that the gauge slice
$\phi=0$ intersects it at $\theta=0$. To calculate
\begin{gather*}
    \langle \exp(\frac{i}{h}S_{int}) g \rangle_{B,P}
\end{gather*}
up to $h^1$ order, we write $\tilde r = r - r_0$, and take the first few terms of Taylor series of $f,g,\phi$ at $\tilde r = 0,\theta = 0$.
For simplicity, we assume that the value of $f$ at $\tilde r = 0$ is $0$.
We also set $X^1 = \tilde r,X^2 = \theta,X^3 = \lambda$. Then

\begin{gather*}
    f = \frac 1 2 a_2 \tilde r^2 + \frac 1 6 a_3 \tilde r^3 + \frac 1 {24} a_4 \tilde r^4 + O(\tilde r^5)\\
    g = s_0 + s_1 \tilde r + s_2 \tilde r^2 + O(\tilde r^3)\\
    \phi = b_r \tilde r + b_\theta \theta + \frac 1 2 b_{rr} \tilde r^2 + b_{r\theta} \tilde r \theta
    + \frac 1 2 b_{\theta\theta} \theta^2 \\
    + \frac 1 6 b_{rrr} \tilde r^3 + \frac 1 2 b_{rr\theta} \tilde r^2 \theta
    + \frac 1 2 b_{r\theta\theta} \tilde r \theta^2 + \frac 1 6 b_{\theta\theta\theta} \theta^3 + \text{higher}
\end{gather*}

Therefore
\begin{gather*}
    L = D^r \partial_r \phi + D^\theta \partial_\theta \phi \\ = b_\theta + b_{\theta \theta} \theta + b_{r\theta} \tilde{r}
     + \frac 1 2 b_{rr\theta} \tilde r^2 + b_{r\theta\theta} \tilde r \theta +
    \frac 1 2 b_{\theta\theta\theta} \theta^2 + \text{higher}\\ \frac i h S = \frac
    i h f + c L \bar c\\ = \frac i h (\frac 1 2 a_2 \tilde r^2 + b_r \lambda \tilde{r}
     + b_\theta \lambda \theta) + c \bar c b_\theta + \frac i h S_{int}\\ B = \begin{pmatrix}
        a_2 & 0        & b_r      \\
        0   & 0        & b_\theta \\
        b_r & b_\theta & 0
    \end{pmatrix}\\
    P = b_\theta
\end{gather*}

Taking the inverses of $B$ and $P$ gives the Feynman rules for edges. The
result is the same as in Fig.12 of \cite{Reshetikhin_2010}, up to some
constants.
\begin{gather*}
    \langle \tilde r^2\rangle_{B,P} = ih \frac 1 {a_2},
    \langle\tilde r\theta \rangle_{B,P} = -ih\frac{b_r}{a_2b_\theta},
    \langle\tilde r\lambda\rangle_{B,P} = 0\\
    \langle \theta\theta \rangle_{B,P} = ih\frac{b_r^2}{a_2b_\theta^2},
    \langle \theta\lambda\rangle_{B,P} = ih\frac{1}{b_\theta},
    \langle \lambda\lambda\rangle_{B,P} = 0\\
    \langle c\bar c\rangle_{B,P} = \frac 1 {b_\theta}
\end{gather*}

The interaction vertices are given by $\frac{i}{h}S_{int}$:
\begin{gather*}
    \frac{i}{h}S_{int} = \frac{i}{h}(\frac 1 6 a_3 \tilde r^3 + \frac 1 {24} a_4 \tilde r^4\\
    + \frac 1 2 b_{rr} \lambda\tilde r^2 + b_{r\theta} \lambda\tilde r \theta
    + \frac 1 2 b_{\theta\theta} \lambda\theta^2 \\
    + \frac 1 6 b_{rrr} \lambda\tilde r^3 + \frac 1 2 b_{rr\theta} \lambda\tilde r^2 \theta
    + \frac 1 2 b_{r\theta\theta} \lambda\tilde r \theta^2 + \frac 1 6 b_{\theta\theta\theta} \lambda\theta^3  )\\
    + b_{\theta \theta}  c\bar c\theta + b_{r\theta}  c\bar c\tilde r
    + \frac 1 2 b_{rr\theta}  c\bar c\tilde r^2 + b_{r\theta\theta}  c\bar c\tilde r \theta + \frac 1 2 b_{\theta\theta\theta}  c\bar c\theta^2+ \text{higher}\\
\end{gather*}

They look like:

\tikzset{every picture/.style={line width=0.75pt}} 

\begin{figure}
    \scalebox{0.7}{
        \begin{tikzpicture}[x=0.75pt,y=0.75pt,yscale=-1,xscale=1]
            
            \draw    (149.83,45) -- (181.5,115) ;
            \draw    (181.5,115) -- (197.67,44.08) ;
            \draw    (181.5,115) -- (256.67,47.58) ;
            \draw    (338.33,46.5) -- (370,116.5) ;
            \draw    (370,116.5) -- (367.17,46.58) ;
            \draw    (370,116.5) -- (392.17,47.58) ;
            \draw    (370,116.5) -- (427.17,50.58) ;
            \draw    (76.33,126) -- (108,196) ;
            \draw    (108,196) -- (124.17,125.08) ;
            \draw    (108,196) -- (167.17,127.58) ;
            \draw    (252.83,124) -- (284.5,194) ;
            \draw    (284.5,194) -- (300.67,123.08) ;
            \draw    (284.5,194) -- (343.67,125.58) ;
            \draw    (422.83,125) -- (454.5,195) ;
            \draw    (454.5,195) -- (470.67,124.08) ;
            \draw    (454.5,195) -- (513.67,126.58) ;
            \draw    (46.33,220.17) -- (78,290.17) ;
            \draw    (78,290.17) -- (75.17,220.25) ;
            \draw    (78,290.17) -- (100.17,221.25) ;
            \draw    (78,290.17) -- (135.17,224.25) ;
            \draw    (198.33,219.17) -- (230,289.17) ;
            \draw    (230,289.17) -- (227.17,219.25) ;
            \draw    (230,289.17) -- (252.17,220.25) ;
            \draw    (230,289.17) -- (287.17,223.25) ;
            \draw    (357.83,217.17) -- (389.5,287.17) ;
            \draw    (389.5,287.17) -- (386.67,217.25) ;
            \draw    (389.5,287.17) -- (411.67,218.25) ;
            \draw    (389.5,287.17) -- (442.17,218.92) ;
            \draw    (499.83,215.17) -- (531.5,285.17) ;
            \draw    (531.5,285.17) -- (528.67,215.25) ;
            \draw    (531.5,285.17) -- (553.67,216.25) ;
            \draw    (531.5,285.17) -- (584.17,216.92) ;
            \draw  [dash pattern={on 0.84pt off 2.51pt}]  (97.33,323) -- (129,393) ;
            \draw  [dash pattern={on 0.84pt off 2.51pt}]  (129,393) -- (142.67,323.25) ;
            \draw    (129,393) -- (204.17,325.58) ;
            \draw  [dash pattern={on 0.84pt off 2.51pt}]  (331.33,322) -- (363,392) ;
            \draw  [dash pattern={on 0.84pt off 2.51pt}]  (363,392) -- (376.67,322.25) ;
            \draw    (363,392) -- (438.17,324.58) ;
            \draw  [dash pattern={on 0.84pt off 2.51pt}]  (97.83,426.83) -- (129.5,496.83) ;
            \draw  [dash pattern={on 0.84pt off 2.51pt}]  (129.5,496.83) -- (126.67,426.92) ;
            \draw    (129.5,496.83) -- (151.67,427.92) ;
            \draw    (129.5,496.83) -- (186.67,430.92) ;
            \draw  [dash pattern={on 0.84pt off 2.51pt}]  (256.83,426.33) -- (288.5,496.33) ;
            \draw  [dash pattern={on 0.84pt off 2.51pt}]  (288.5,496.33) -- (285.67,426.42) ;
            \draw    (288.5,496.33) -- (310.67,427.42) ;
            \draw    (288.5,496.33) -- (345.67,430.42) ;
            \draw  [dash pattern={on 0.84pt off 2.51pt}]  (450.83,424.83) -- (482.5,494.83) ;
            \draw  [dash pattern={on 0.84pt off 2.51pt}]  (482.5,494.83) -- (479.67,424.92) ;
            \draw    (482.5,494.83) -- (504.67,425.92) ;
            \draw    (482.5,494.83) -- (539.67,428.92) ;
            
            \draw (160,48.4) node [anchor=north west][inner sep=0.75pt]    {$r$};
            \draw (199.67,47.48) node [anchor=north west][inner sep=0.75pt]    {$r$};
            \draw (258,44.9) node [anchor=north west][inner sep=0.75pt]    {$r$};
            \draw (112.5,56.4) node [anchor=north west][inner sep=0.75pt]    {$\frac{i}{h} a_{3}$};
            \draw (348.5,49.9) node [anchor=north west][inner sep=0.75pt]    {$r$};
            \draw (370.67,47.48) node [anchor=north west][inner sep=0.75pt]    {$r$};
            \draw (394.5,47.9) node [anchor=north west][inner sep=0.75pt]    {$r$};
            \draw (301,57.9) node [anchor=north west][inner sep=0.75pt]    {$\frac{i}{h} a_{4}$};
            \draw (433,47.4) node [anchor=north west][inner sep=0.75pt]    {$r$};
            \draw (86.5,129.4) node [anchor=north west][inner sep=0.75pt]    {$\lambda $};
            \draw (126.17,128.48) node [anchor=north west][inner sep=0.75pt]    {$r$};
            \draw (170.5,125.4) node [anchor=north west][inner sep=0.75pt]    {$r$};
            \draw (39,137.4) node [anchor=north west][inner sep=0.75pt]    {$\frac{i}{h} b_{rr}$};
            \draw (263,127.4) node [anchor=north west][inner sep=0.75pt]    {$\lambda $};
            \draw (302.67,126.48) node [anchor=north west][inner sep=0.75pt]    {$r$};
            \draw (347,123.4) node [anchor=north west][inner sep=0.75pt]    {$\theta $};
            \draw (215.5,135.4) node [anchor=north west][inner sep=0.75pt]    {$\frac{i}{h} b_{r\theta }$};
            \draw (433,128.4) node [anchor=north west][inner sep=0.75pt]    {$\lambda $};
            \draw (385.5,136.4) node [anchor=north west][inner sep=0.75pt]    {$\frac{i}{h} b_{\theta \theta }$};
            \draw (475.5,125.9) node [anchor=north west][inner sep=0.75pt]    {$\theta $};
            \draw (518,126.4) node [anchor=north west][inner sep=0.75pt]    {$\theta $};
            \draw (78.67,221.15) node [anchor=north west][inner sep=0.75pt]    {$r$};
            \draw (102.5,221.57) node [anchor=north west][inner sep=0.75pt]    {$r$};
            \draw (9,231.57) node [anchor=north west][inner sep=0.75pt]    {$\frac{i}{h} b_{rrr}$};
            \draw (141,221.07) node [anchor=north west][inner sep=0.75pt]    {$r$};
            \draw (55.5,220.4) node [anchor=north west][inner sep=0.75pt]    {$\lambda $};
            \draw (230.67,220.15) node [anchor=north west][inner sep=0.75pt]    {$r$};
            \draw (254.5,220.57) node [anchor=north west][inner sep=0.75pt]    {$r$};
            \draw (161.5,230.57) node [anchor=north west][inner sep=0.75pt]    {$\frac{i}{h} b_{rr\theta }$};
            \draw (293,220.07) node [anchor=north west][inner sep=0.75pt]    {$\theta $};
            \draw (207.5,219.4) node [anchor=north west][inner sep=0.75pt]    {$\lambda $};
            \draw (390.17,218.15) node [anchor=north west][inner sep=0.75pt]    {$r$};
            \draw (320.5,228.57) node [anchor=north west][inner sep=0.75pt]    {$\frac{i}{h} b_{r\theta \theta }$};
            \draw (444.5,217.57) node [anchor=north west][inner sep=0.75pt]    {$\theta $};
            \draw (367,217.4) node [anchor=north west][inner sep=0.75pt]    {$\lambda $};
            \draw (413,218.07) node [anchor=north west][inner sep=0.75pt]    {$\theta $};
            \draw (462.5,226.57) node [anchor=north west][inner sep=0.75pt]    {$\frac{i}{h} b_{\theta \theta \theta }$};
            \draw (586.5,215.57) node [anchor=north west][inner sep=0.75pt]    {$\theta $};
            \draw (509,215.4) node [anchor=north west][inner sep=0.75pt]    {$\lambda $};
            \draw (555,216.07) node [anchor=north west][inner sep=0.75pt]    {$\theta $};
            \draw (532,216.07) node [anchor=north west][inner sep=0.75pt]    {$\theta $};
            \draw (107.5,326.4) node [anchor=north west][inner sep=0.75pt]    {$c$};
            \draw (142.17,325.48) node [anchor=north west][inner sep=0.75pt]    {$\overline{c}$};
            \draw (205.5,322.9) node [anchor=north west][inner sep=0.75pt]    {$r$};
            \draw (60,334.4) node [anchor=north west][inner sep=0.75pt]    {$b_{r\theta }$};
            \draw (341.5,325.4) node [anchor=north west][inner sep=0.75pt]    {$c$};
            \draw (376.17,324.48) node [anchor=north west][inner sep=0.75pt]    {$\overline{c}$};
            \draw (439.5,321.9) node [anchor=north west][inner sep=0.75pt]    {$\theta $};
            \draw (294,333.4) node [anchor=north west][inner sep=0.75pt]    {$b_{\theta \theta }$};
            \draw (154,428.23) node [anchor=north west][inner sep=0.75pt]    {$r$};
            \draw (61,438.23) node [anchor=north west][inner sep=0.75pt]    {$b_{rr\theta }$};
            \draw (192.5,427.73) node [anchor=north west][inner sep=0.75pt]    {$r$};
            \draw (220,436.23) node [anchor=north west][inner sep=0.75pt]    {$b_{r\theta \theta }$};
            \draw (399.5,434.23) node [anchor=north west][inner sep=0.75pt]    {$b_{\theta \theta \theta }$};
            \draw (106.5,425.4) node [anchor=north west][inner sep=0.75pt]    {$c$};
            \draw (131.17,426.48) node [anchor=north west][inner sep=0.75pt]    {$\overline{c}$};
            \draw (313,427.73) node [anchor=north west][inner sep=0.75pt]    {$r$};
            \draw (351.5,427.23) node [anchor=north west][inner sep=0.75pt]    {$\theta $};
            \draw (265.5,424.9) node [anchor=north west][inner sep=0.75pt]    {$c$};
            \draw (290.17,425.98) node [anchor=north west][inner sep=0.75pt]    {$\overline{c}$};
            \draw (507,426.23) node [anchor=north west][inner sep=0.75pt]    {$\theta $};
            \draw (545.5,425.73) node [anchor=north west][inner sep=0.75pt]    {$\theta $};
            \draw (459.5,423.4) node [anchor=north west][inner sep=0.75pt]    {$c$};
            \draw (484.17,424.48) node [anchor=north west][inner sep=0.75pt]    {$\overline{c}$};

            \end{tikzpicture}}    
            \caption{Interaction Vertices}
\end{figure}
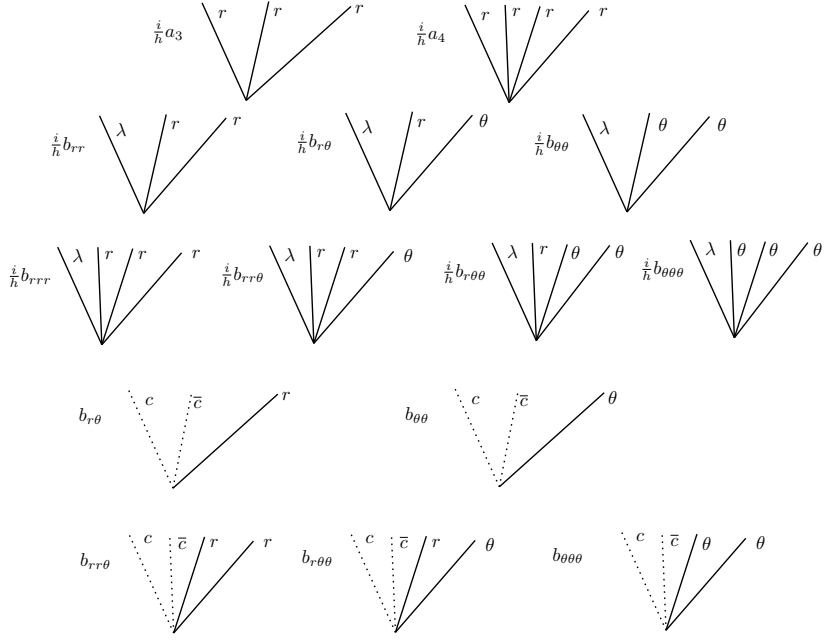

To calculate through order $h^1$, it suffices to consider the fourteen vertices
above into consideration, which is the same as in Fig.13 of
\cite{Reshetikhin_2010}, up to some constants.

We now deform $\phi$, equivalently its Taylor coefficients. This changes
Taylor series of $\phi$. This will cause the deformation of $B,P$ and
$S_{int}$, as we stated in proposition \ref{deformation of expectation}. In
fact, Proposition~\ref{deformation of expectation} expresses this variation
diagrammatically as a correspondence between Feynman diagrams.

Because there are many diagrams, we illustrate the correspondence with one
explicit example. In the expansion of $\langle
    \exp(\frac{i}{h}S_{int})g\rangle_{B,P}$, there is a term
\begin{gather*}
    \langle \frac{i}{h}b_{r\theta}\lambda\tilde r \theta\cdot b_{\theta\theta} c \bar c \theta \cdot s_0\rangle_{B,P}
\end{gather*}

This term is the sum of three Feynman diagrams. One of them is

\begin{figure}[h]
    \scalebox{1.2}{
        \begin{tikzpicture}[x=0.75pt,y=0.75pt,yscale=-1,xscale=1]
            
            \draw   (58,83.23) .. controls (58,64.92) and (72.85,50.07) .. (91.17,50.07) .. controls (109.48,50.07) and (124.33,64.92) .. (124.33,83.23) .. controls (124.33,101.55) and (109.48,116.4) .. (91.17,116.4) .. controls (72.85,116.4) and (58,101.55) .. (58,83.23) -- cycle ;
            \draw  [dash pattern={on 0.84pt off 2.51pt}] (182.4,83.43) .. controls (182.4,65.12) and (197.25,50.27) .. (215.57,50.27) .. controls (233.88,50.27) and (248.73,65.12) .. (248.73,83.43) .. controls (248.73,101.75) and (233.88,116.6) .. (215.57,116.6) .. controls (197.25,116.6) and (182.4,101.75) .. (182.4,83.43) -- cycle ;
            \draw    (124.33,83.23) -- (182.4,83.43) ;
            
            \draw (128.8,72) node [anchor=north west][inner sep=0.75pt]  [font=\scriptsize]  {$\lambda $};
            \draw (111.6,64.4) node [anchor=north west][inner sep=0.75pt]  [font=\scriptsize]  {$r$};
            \draw (110,92.2) node [anchor=north west][inner sep=0.75pt]  [font=\scriptsize]  {$\theta $};
            \draw (167.67,71.87) node [anchor=north west][inner sep=0.75pt]  [font=\scriptsize]  {$\theta $};
            \draw (191.67,94.53) node [anchor=north west][inner sep=0.75pt]  [font=\scriptsize]  {$\overline{c}$};
            \draw (193,61.2) node [anchor=north west][inner sep=0.75pt]  [font=\scriptsize]  {$c$};

            \end{tikzpicture}}    
            \caption{Example Feynman diagram}
\end{figure}
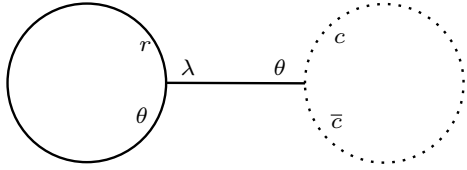

This diagram contributes
\begin{gather*}
    \frac{i}{h}b_{r\theta} b_{\theta\theta} \langle\tilde r\theta \rangle_{B,P}\langle\lambda \theta \rangle_{B,P} \langle c \bar c \rangle_{B,P}
\end{gather*}

Now differentiate with respect to $t$ under the hypotheses of
Theorem~\ref{deformation theorem 1}. For simplicity, assume that $B$ and $P$
are constant in $t$. This gives two terms:
\begin{gather*}
    \frac{i}{h}\frac{d}{dt}b_{r\theta} b_{\theta\theta} \langle\tilde r\theta \rangle_{B,P}\langle\lambda \theta \rangle_{B,P} \langle c \bar c \rangle_{B,P}\\
    +\frac{i}{h}b_{r\theta}\frac{d}{dt} b_{\theta\theta} \langle\tilde r\theta \rangle_{B,P}\langle\lambda \theta \rangle_{B,P} \langle c \bar c \rangle_{B,P}
\end{gather*}

Each term can be represented by a Feynman diagram:

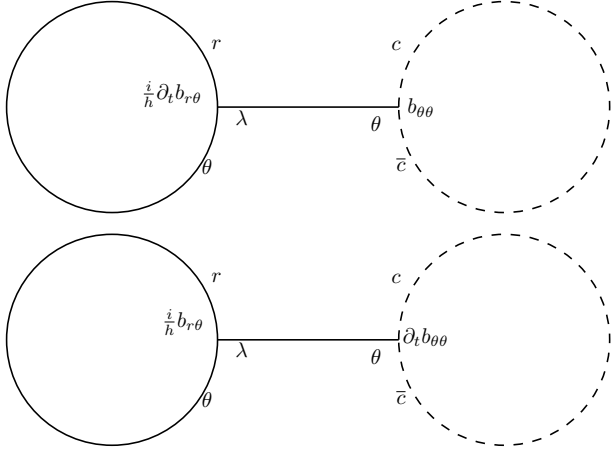
\begin{figure}[h]
    \scalebox{0.8}{
        \begin{tikzpicture}[x=0.75pt,y=0.75pt,yscale=-1,xscale=1]
            
            \draw   (97,107.17) .. controls (97,70.62) and (126.62,41) .. (163.17,41) .. controls (199.71,41) and (229.33,70.62) .. (229.33,107.17) .. controls (229.33,143.71) and (199.71,173.33) .. (163.17,173.33) .. controls (126.62,173.33) and (97,143.71) .. (97,107.17) -- cycle ;
            \draw  [dash pattern={on 4.5pt off 4.5pt}] (343,107.17) .. controls (343,70.62) and (372.62,41) .. (409.17,41) .. controls (445.71,41) and (475.33,70.62) .. (475.33,107.17) .. controls (475.33,143.71) and (445.71,173.33) .. (409.17,173.33) .. controls (372.62,173.33) and (343,143.71) .. (343,107.17) -- cycle ;
            \draw    (229.33,107.17) -- (343,107.17) ;
            \draw   (97,253.17) .. controls (97,216.62) and (126.62,187) .. (163.17,187) .. controls (199.71,187) and (229.33,216.62) .. (229.33,253.17) .. controls (229.33,289.71) and (199.71,319.33) .. (163.17,319.33) .. controls (126.62,319.33) and (97,289.71) .. (97,253.17) -- cycle ;
            \draw  [dash pattern={on 4.5pt off 4.5pt}] (343,253.17) .. controls (343,216.62) and (372.62,187) .. (409.17,187) .. controls (445.71,187) and (475.33,216.62) .. (475.33,253.17) .. controls (475.33,289.71) and (445.71,319.33) .. (409.17,319.33) .. controls (372.62,319.33) and (343,289.71) .. (343,253.17) -- cycle ;
            \draw    (229.33,253.17) -- (343,253.17) ;
            
            \draw (224,63.4) node [anchor=north west][inner sep=0.75pt]    {$r$};
            \draw (239,107.4) node [anchor=north west][inner sep=0.75pt]    {$\lambda $};
            \draw (218,138.4) node [anchor=north west][inner sep=0.75pt]    {$\theta $};
            \draw (180,89.4) node [anchor=north west][inner sep=0.75pt]    {$\frac{i}{h} \partial _{t} b_{r\theta }$};
            \draw (337,64.4) node [anchor=north west][inner sep=0.75pt]    {$c$};
            \draw (340,137.4) node [anchor=north west][inner sep=0.75pt]    {$\overline{c}$};
            \draw (324,111.4) node [anchor=north west][inner sep=0.75pt]    {$\theta $};
            \draw (347,100.4) node [anchor=north west][inner sep=0.75pt]    {$b_{\theta \theta }$};
            \draw (224,209.4) node [anchor=north west][inner sep=0.75pt]    {$r$};
            \draw (239,253.4) node [anchor=north west][inner sep=0.75pt]    {$\lambda $};
            \draw (218,284.4) node [anchor=north west][inner sep=0.75pt]    {$\theta $};
            \draw (193,232.4) node [anchor=north west][inner sep=0.75pt]    {$\frac{i}{h} b_{r\theta }$};
            \draw (337,210.4) node [anchor=north west][inner sep=0.75pt]    {$c$};
            \draw (340,283.4) node [anchor=north west][inner sep=0.75pt]    {$\overline{c}$};
            \draw (324,257.4) node [anchor=north west][inner sep=0.75pt]    {$\theta $};
            \draw (344,244.4) node [anchor=north west][inner sep=0.75pt]    {$\partial _{t} b_{\theta \theta }$};

            \end{tikzpicture}}    
            \caption{Derivative of the example Feynman diagram}
\end{figure}

These Feynman diagrams contribute to
\begin{gather*}
    \langle \exp(\frac{i}{h}S_{int})\frac{i}{h}\frac{\partial}{\partial t}S_{int} \cdot g\rangle_{B,P}
\end{gather*}

A term-by-term correspondence of Feynman diagrams gives, in this example,
\begin{gather*}
    \frac{\partial }{\partial t}\langle \exp(\frac{i}{h}S_{int})g\rangle_{B,P} =\langle
    \exp(\frac{i}{h}S_{int})\frac{i}{h}\frac{\partial}{\partial t}S_{int} \cdot g\rangle_{B,P}
\end{gather*}

This agrees with Proposition~\ref{deformation of expectation}. The next step is
to show that the right-hand side vanishes. This uses two Feynman-diagram
identities corresponding to Lemmas~\ref{Q_C lemma} and~\ref{Q_D lemma}.

Lemma \ref{Q_D lemma} implies
\begin{gather*}
    \langle  \exp(\frac{i}{h}S_{int})(-\frac{i}{h}c\bar c \lambda \partial_\theta \phi \cdot \frac{\partial }{\partial t}\phi + \frac{i}{h}\lambda \frac{\partial }{\partial t}\phi)g \rangle_{B,P} = 0
\end{gather*}

This can be interpreted as a Feynman-diagram identity. Since $\langle c
    \bar c \rangle_{B,P} = b_\theta^{-1}$, so, whenever there exists a $c \bar c$
loop in the diagram, we can delete it and multiply the diagram with
$b_\theta^{-1}$, which will keep the value of the diagram invariant.

Apply this operation to each diagram in the expansion of
\begin{gather*}
    \langle  \exp(\frac{i}{h}S_{int})\frac{i}{h}c\bar c \lambda b_\theta  \cdot \frac{\partial }{\partial t}\phi g\rangle_{B,P}
\end{gather*}

We can get part of the diagrams in the expansion of
\begin{gather*}
    \langle  \exp(\frac{i}{h}S_{int})\frac{i}{h} \lambda \cdot \frac{\partial }{\partial t}\phi g\rangle_{B,P}
\end{gather*}

The remaining diagrams in the above expansion comes from the $c \bar c$ terms
in the $\frac{i}{h}S_{int}$: This term is $c\bar c (\partial_\theta\phi - b_\theta)$, so the remaining diagrams can be obtained by
\begin{gather*}
    \langle \exp(\frac{i}{h}S_{int})(\frac{i}{h}c\bar c \lambda (\partial_\theta\phi - b_\theta) \cdot \frac{\partial }{\partial t}\phi)g \rangle_{B,P}
\end{gather*}

Combining the contributions gives
\begin{gather*}
    \langle  \exp(\frac{i}{h}S_{int})\frac{i}{h}c\bar c \lambda b_\theta  \cdot \frac{\partial }{\partial t}\phi g\rangle_{B,P} =\langle \exp(\frac{i}{h}S_{int})\frac{i}{h} \lambda
    \cdot \frac{\partial }{\partial t}\phi g\rangle_{B,P} -\langle \exp(\frac{i}{h}S_{int})(\frac{i}{h}c\bar c
    \lambda (\partial_\theta\phi - b_\theta) \cdot \frac{\partial }{\partial t}\phi)g \rangle_{B,P}
\end{gather*}

This recovers what we want.

Lemma \ref{Q_C lemma} implies
\begin{gather*}
    \langle  \exp(\frac{i}{h}S_{int})(\frac{i}{h}c\bar c \lambda \partial_\theta \phi \cdot \frac{\partial }{\partial t}\phi + c\bar c \frac{\partial }{\partial t}\partial_\theta\phi)g \rangle_{B,P} = 0
\end{gather*}

This can be interpreted as a Feynman-diagram identity. First consider all
Feynman diagrams in the expansion of
\begin{gather*}
    \langle  \exp(\frac{i}{h}S_{int})(-\frac{i}{h}c \lambda b_\theta \cdot \bar c\frac{\partial }{\partial t}\phi)g \rangle_{B,P}
\end{gather*}

Focus on the vertex associated with $-\frac{i}{h}c\lambda b_\theta$: $c$ must
be in an edge to $\bar c$ somewhere, and $\lambda$ must be in an edge to
$\lambda,r$ or $\theta$. This edge is nonzero only for the last case. In that
case, we remove $-\frac{i}{h}c \lambda b_\theta$ and the $\theta$, and then
replace it by $c$, and then connect it with the $\bar c$ mentioned above. This
operation still keeps the value of the diagram invariant, since $\langle
    \theta\lambda\rangle_{B,P} = ih\frac{1}{b_\theta}$. The operation is shown in
the figure 4 below.

\begin{figure}[ht]
    \scalebox{1.0}{
        \begin{tikzpicture}[x=0.75pt,y=0.75pt,yscale=-1,xscale=1]
            
            \draw    (123.33,62.17) -- (158.33,152.17) ;
            \draw  [dash pattern={on 0.84pt off 2.51pt}]  (158.33,152.17) -- (201.33,60.17) ;
            \draw    (225,107) .. controls (268.89,91.32) and (269.66,118.77) .. (312.04,104.61) ;
            \draw [shift={(313.33,104.17)}, rotate = 160.85] [color={rgb, 255:red, 0; green, 0; blue, 0 }  ][line width=0.75]    (10.93,-3.29) .. controls (6.95,-1.4) and (3.31,-0.3) .. (0,0) .. controls (3.31,0.3) and (6.95,1.4) .. (10.93,3.29)   ;
            \draw  [dash pattern={on 0.84pt off 2.51pt}]  (413,91) .. controls (438.33,114.17) and (488.33,117.17) .. (512.33,89.17) ;
            
            \draw (114,68.4) node [anchor=north west][inner sep=0.75pt]    {$\theta $};
            \draw (134,140.4) node [anchor=north west][inner sep=0.75pt]    {$\lambda $};
            \draw (173,138.4) node [anchor=north west][inner sep=0.75pt]    {$c$};
            \draw (209,60.4) node [anchor=north west][inner sep=0.75pt]    {$\overline{c}$};
            \draw (403,98.4) node [anchor=north west][inner sep=0.75pt]    {$c$};
            \draw (515,100.4) node [anchor=north west][inner sep=0.75pt]    {$\overline{c}$};
            \draw (339,91.4) node [anchor=north west][inner sep=0.75pt]    {$ih\cdot \frac{1}{b_{\theta }}$};

            \end{tikzpicture}}    
            \caption{Feynman-diagram operation}\label{Feynman diagram operation}
\end{figure}
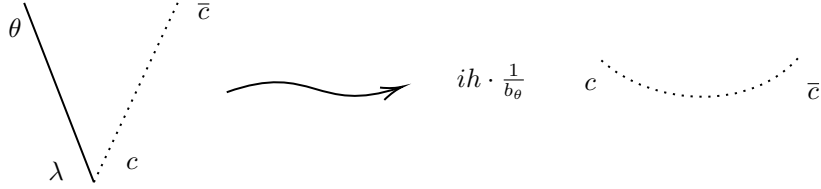

As we apply this operation, the $\theta$ can be from 2 places: the vertices
$\bar c\frac{\partial }{\partial t}\phi g$, or the vertices in $\frac{i}{h}S_{int}$.

In the first case, we obtain a Feynman diagram in the expansion of the following
expression (recall that $\partial_\theta g=0$):
\begin{gather*}
    \langle  \exp(\frac{i}{h}S_{int})c\bar c\frac{\partial }{\partial t}(\partial_\theta \phi) g \rangle_{B,P}
\end{gather*}

In the second case, we obtain a Feynman diagram in the expansion of
\begin{gather*}
    \langle  \exp(\frac{i}{h}S_{int})(\frac{i}{h}c \lambda (\partial_\theta \phi - b_\theta)\cdot \bar c \frac{\partial }{\partial t}\phi )g \rangle_{B,P}
\end{gather*}

Combining the two cases gives
\begin{gather*}
    \langle  \exp(\frac{i}{h}S_{int})(-\frac{i}{h}c \lambda b_\theta \cdot \bar c\frac{\partial }{\partial t}\phi)g \rangle_{B,P} =\langle \exp(\frac{i}{h}S_{int})c\bar c\frac{\partial }{\partial t}(\partial_\theta \phi) g \rangle_{B,P} +\langle
    \exp(\frac{i}{h}S_{int})(\frac{i}{h}c \lambda (\partial_\theta \phi - b_\theta)\cdot \bar c \frac{\partial }{\partial t}\phi )g \rangle_{B,P}
\end{gather*}

This recovers what we want.

Now, we take the sum of the results of the 2 lemmas, we get
\begin{gather*}
    0 = \langle  \exp(\frac{i}{h}S_{int})\frac{i}{h} \lambda \cdot \frac{\partial }{\partial t}\phi g\rangle_{B,P} +\langle \exp(\frac{i}{h}S_{int})c\bar c\frac{\partial }{\partial t}(\partial_\theta \phi) g \rangle_{B,P}\\ =\langle
    \exp(\frac{i}{h}S_{int})\frac{i}{h}\frac{\partial}{\partial t}S_{int} \cdot g\rangle_{B,P}\\ =\frac{\partial }{\partial t}\langle \exp(\frac{i}{h}S_{int})g\rangle_{B,P}
\end{gather*}

The last step also follows directly from the Feynman-diagram
diagrams.

In the non-Abelian case, the diagrammatic identities are more complicated, but
they are obtained in the same way from Lemmas~\ref{Q_C lemma}
and~\ref{Q_D lemma}. We omit this diagrammatic reformulation because the
algebraic proof above already covers that case.

\section{Appendix: Independence of the Choice of Coordinates}\label{Appendix : Independence of the Choice of Coordinates}

One question remains: is the algebraically defined asymptotic expansion
independent of the choice of coordinates?

We answer this question by showing that the normalized expansion is invariant
under a smooth deformation of coordinates that preserves the hypotheses.

Under the coordinate deformation generated by $W^k\partial_k$, the functions,
density, and $\mathfrak g$-action vary as follows:
\begin{gather*}
    \frac{\partial}{\partial t} f = W^k\partial_k f, \quad \frac{\partial}{\partial t} g =\partial_k ( W^k g), \quad \frac{\partial}{\partial t} D_a^l = W^k\partial_k D_a^l - D_a^k\partial_k W^l, \quad \frac{\partial}{\partial t} \phi^a = W^k\partial_k \phi^a
\end{gather*}

This is stronger than the condition given in the theorem
\ref{algebraicindependence}, so we get
\begin{gather*}
    \frac{\partial}{\partial t}(\frac{\det P}{\sqrt{\det B}} \langle \exp(\frac{i}{h}S_{int}) g
    \rangle_{B,P}) = 0
\end{gather*}

\section{Appendix: Formal Integral Interpretation of the Lemmas and Proofs}\label{Appendix: Integral Interpretation of the Lemmas and Proofs}
This appendix gives a formal integral interpretation of the algebraic proof.
In field theory the analogous expressions are formal path integrals; no
analytic measure-theoretic claim is made here. For
simplicity, we write $DxD\lambda DcD\bar c =
    dX^1...dX^{n+m}dc^1...dc^md\bar{c}_1...d\bar{c}_m$, and $L_a^b = D_a^k\partial_k\phi^b$.

The Berezin integral of a total derivative with respect to any $\bar c_a$
vanishes. Since $Q_D=M_a\overleftarrow\partial/\partial\bar c_a$, this gives
\begin{gather*}
    0 = \int Q_D(\exp(\frac i h S)\bar c _a \frac{\partial }{\partial t}\phi^a)g DxD\lambda DcD\bar c\\ =\int \exp(\frac i h S)(Q_D(\frac i h
    S)\bar c _a \frac{\partial }{\partial t}\phi^a + Q_D(\bar c_a\frac{\partial }{\partial t}\phi^a))g DxD\lambda DcD\bar c\\ =\int \exp(\frac i h S)(-\frac i h
    \lambda_b c^c D_c^k\partial_k\phi^b \cdot \bar c _a \frac{\partial }{\partial t}\phi^a + \frac i h \lambda_a\frac{\partial }{\partial t}\phi^a )g DxD\lambda DcD\bar c
\end{gather*}

When $g$ is a $\mathfrak g$-invariant density, i.e., for all $a$,
\begin{gather*}
    \partial_k(D_a^k g) = 0
\end{gather*}

In this case one may similarly argue, formally, that for any $\alpha$,
\begin{gather}\label{general Q_C identity}
    0 = \int Q_C(\alpha)g DxD\lambda DcD\bar c
\end{gather}

Therefore we will get,
\begin{gather}\label{special Q_C identity}
    \begin{split}
        0 =& \int Q_C(\exp(\frac i h S)\bar c _a \frac{\partial }{\partial t}\phi^a)g DxD\lambda DcD\bar c\\ =&\int \exp(\frac i h S)(Q_C(\frac i h
        S)\bar c _a \frac{\partial }{\partial t}\phi^a + Q_C(\bar c_a\frac{\partial }{\partial t}\phi^a))g DxD\lambda DcD\bar c\\ =&\int \exp(\frac i h S)(\frac i h
        \lambda_b c^c D_c^k\partial_k\phi^b \cdot \bar c _a \frac{\partial }{\partial t}\phi^a + c^a \bar c_b D_a^k \partial_k\frac{\partial \phi^b}{\partial t})g DxD\lambda DcD\bar c
    \end{split}
\end{gather}

Adding the $Q_C$ and $Q_D$ identities gives
\begin{gather*}
    0 = \int \exp(\frac i h S)(c^a \bar c_b D_a^k \partial_k\frac{\partial \phi^b}{\partial t} + \frac i h \lambda_a\frac{\partial \phi^a}{\partial t} )g DxD\lambda DcD\bar c\\ =\frac{\partial}{\partial t}\int \exp(\frac i h S) g DxD\lambda DcD\bar c
\end{gather*}

That is what we want.

For completeness, we now verify~\eqref{special Q_C identity} directly rather
than proving the general identity~\eqref{general Q_C identity}.

Formal integration by parts gives

\begin{gather*}
    \int \exp(\frac i h S)( c^a \bar c _b D_a^k\partial_k\frac{\partial \phi^b}{\partial t})g DxD\lambda DcD\bar c\\ =\int (- c^a \bar c _b \frac{\partial \phi^b}{\partial t})\partial_k(\exp(\frac i h S) g D_a^k)DxD\lambda DcD\bar c\\ =\int \bar c_b \frac{\partial \phi^b}{\partial t} g c^a D_a^k \partial_k \exp(\frac i h S)DxD\lambda DcD\bar c\\
\end{gather*}

Notice that
\begin{gather*}
    c^a D_a^k\partial_k \left(\frac i h S\right)
    =c^a D_a^k\partial_k\left(\frac i h f + c^b\bar c_c D_b^l \partial_l \phi^c + \frac i h \lambda_b\phi^b \right)\\
    =c^a c^b\bar c_c D_a^k\partial_k\left(D_b^l\partial_l\phi^c\right)
    + \frac i h \lambda_b c^a D_a^k\partial_k \phi^b.
\end{gather*}

Because $c^ac^b=-c^bc^a$, only the antisymmetric part of the first term
contributes. Using $[D_a,D_b]=f_{ab}^dD_d$ and the invariance of $f$, we obtain
\begin{gather*}
    c^a D_a^k\partial_k (\frac i h S)\\ =c^a c^b \bar c _c \frac 1 2 f_{ab}^d D_d^k\partial_k \phi^c + \frac i h \lambda_b c^a D_a^k\partial_k \phi^b
\end{gather*}

So
\begin{gather*}
    \int  \bar c_b  \frac{\partial }{\partial t}\phi^b g c^a D_a^k \partial_k \exp(\frac i h S)DxD\lambda DcD\bar c\\ =\int \bar c_b \frac{\partial }{\partial t}\phi^b g \exp(\frac i h S)(c^a c^e \bar c _c \frac 1 2 f_{ae}^d D_d^k\partial_k \phi^c + \frac i h \lambda_e c^a D_a^k\partial_k \phi^e)DxD\lambda DcD\bar c\\
\end{gather*}

Write $L_a^b=D_a^k\partial_k\phi^b$. Since the ghost-dependent part of the
exponent is $c^a\bar c_bL_a^b$, the Berezin integral satisfies
\begin{gather*}
    \int \exp(c^a \bar c_b L_a^b)( c^c \bar c_d) Dc D\bar c = \int \exp(c^a \bar c_b L_a^b) Dc D\bar c \cdot (L^{-1})^c_d\\
    \int \exp(c^a \bar c_b L_a^b)( c^c \bar c_d c^e\bar c_f) Dc D\bar c = \int \exp(c^a \bar c_b L_a^b) Dc D\bar c \cdot ((L^{-1})^c_d (L^{-1})^e_f -  (L^{-1})^c_f (L^{-1})^e_d)\\
\end{gather*}

So
\begin{gather*}
    \int \bar c_b  \frac{\partial }{\partial t}\phi^b g \exp(\frac i h S)(c^a c^e \bar c _c \frac 1 2 f_{ae}^d D_d^k\partial_k \phi^c )DxD\lambda DcD\bar c\\ =\int \exp(\frac i h S) \bar c_b c^a c^e
    \bar c _c \frac 1 2 f_{ae}^d L_d^c \frac{\partial }{\partial t}\phi^b g DxD\lambda DcD\bar c\\ =\int \exp(\frac i h S)(-1)((L^{-1})_b^a
    (L^{-1})_c^e - (L^{-1})_b^e (L^{-1})_c^a)\cdot \frac 1 2 f_{ae}^d L_d^c \frac{\partial }{\partial t}\phi^b g DxD\lambda DcD\bar c\\ =\int \exp(\frac i h
    S)(-\frac{1}{2}(L^{-1})_b^a f_{ad}^d + \frac{1}{2} (L^{-1})_b^e f^a_{ae})
    \frac{\partial }{\partial t}\phi^b g DxD\lambda DcD\bar c\\ =0
\end{gather*}

The last equality uses the unimodularity of $\mathfrak g$, namely
$f_{ad}^{d}=f_{ae}^{a}=0$ (after relabeling dummy indices).

In total, we get
\begin{gather*}
    \int \exp(\frac i h S)( c^a \bar c _b D_a^k\partial_k\frac{\partial \phi^b}{\partial t})g DxD\lambda DcD\bar c\\ =\int \bar c_b \frac{\partial \phi^b}{\partial t} g \exp(\frac i h S)(c^a c^e \bar c _c \frac 1 2 f_{ae}^d D_d^k\partial_k \phi^c + \frac i h \lambda_e c^a D_a^k\partial_k \phi^e)DxD\lambda DcD\bar c\\ =\int \bar c_b \frac{\partial \phi^b}{\partial t} g \exp(\frac i h S)( \frac i h \lambda_e c^a D_a^k\partial_k \phi^e)DxD\lambda DcD\bar c\\ =-\int \exp(\frac i h S)\frac i h
    \lambda_b c^c D_c^k\partial_k\phi^b \cdot \bar c _a \frac{\partial }{\partial t}\phi^a g DxD\lambda DcD\bar c
\end{gather*}

That is just \ref{special Q_C identity}.

\bibliographystyle{unsrt}

\bibliography{ref.bib}

\end{document}